\pdfoutput=1
\documentclass[runningheads,a4paper]{llncs}
\usepackage{amsmath}
\usepackage{amssymb}
 \usepackage{listings}  
\usepackage{stmaryrd}
\usepackage{graphicx}
\usepackage{euscript}
\usepackage{bbm}
\usepackage{paralist}
\usepackage{color}
\usepackage{mathbbol}
\usepackage{mathabx}
\usepackage{paralist}
\usepackage{rotating}
\usepackage{url}
\lstnewenvironment{code}{%
  \lstset{frame=single,
   basicstyle=\ttfamily,
  mathescape,
  escapeinside=||,}
}{}

\newcommand{\bigfract}[2]{\frac{^{\textstyle #1}}{_{\textstyle #2}}}
\newcommand{\rulename}[1]{{\sc (#1)}}
\newcommand{\rulenamex}[1]{\mbox{\scriptsize\sc(#1)}}

\def \mathrule #1#2#3{\begin{array}{l} 
                       {\rulenamex{#1}}
                       \\ \bigfract{#2}{#3}
                      \end{array}}

\newcommand{\subst}[2]{[\raisebox{.5ex}{\footnotesize$#1$}  /
                        \raisebox{-.5ex}{\footnotesize$#2$} ]}

\newcommand{\M}{\mathtt{f}}
\newcommand{\Q}{\mathtt{l}}
\newcommand{\PP}{\mathtt{h}}
\newcommand{\N}{\mathtt{g}}

\newcommand{\CP}{\mathfrak{L}}
\renewcommand{\S}{\mathtt{S}}
\renewcommand{\P}{\mathtt{L}}
\newcommand{\Mplus}{\mathbb{F}}

\newcommand{\f}{{\tt f}}

\newcommand{\seq}{\raisebox{.2ex}{\small ~{\bf +}~}}
\newcommand{\sparop}{\binampersand}

\renewcommand{\prod}{\sparop}

\newcommand{\transclosure}{\oplus}
\newcommand{\lessc}{{\scriptstyle <}}

\newcommand{\ordermut}[1]{\mathfrak{o}_{#1}}
\newcommand{\suh}[1]{{}^\mathfrak{h}{#1}}
\newcommand{\suhs}[2]{{}^{#1}{#2}}
\newcommand{\addh}[2]{{\it addh}(#1,#2)}
\newcommand{\lcm}{{\it lcm}}
\newcommand{\lessfb}{\preceq^{{\tt fb}}}

\newcommand{\flatt}[1]{\flat(#1)}
\newcommand{\flatm}[2]{\flat_{#1}(#2)}

\newcommand{\lred}[1]{\stackrel{#1}{\longrightarrow}}

\newcommand{\pslinearized}{\Supset_{\tt pl}}
\newcommand{\linearized}{\Supset_{\tt lin}}

\newcommand{\rechis}[1]{{\it rechis}(#1)}
\newcommand{\head}[1]{{\it head}(#1)}

\newcommand{\transfuno}[2]{#1 \; \stackrel{{\tt pl} \mapsto {\tt l}}{\Longmapsto}_1 \; #2}
\newcommand{\transfdue}[2]{#1 \; \stackrel{{\tt pl} \mapsto {\tt l}}{\Longmapsto}_2 \; #2}

\newcommand{\transflam}[1]{\stackrel{#1}{\Longmapsto}}

\newcommand{\T}{{\tt T}}

\newcommand{\por}[1]{\mathbb{#1}}

\newcommand{\lam}[1]{\bigl\langle #1 \bigr\rangle}
\newcommand{\ilam}[2]{\bigl\langle #1 \bigr\rangle_{#2}}

\newcommand{\dom}[1]{{\it dom}(#1)}

\newcommand{\closure}[1]{{\it closure}(#1)}

\newcommand{\obj}{x}
\newcommand{\objb}{y}

\newcommand{\wt}[1]{\widetilde{#1}}
\newcommand{\pinull}{{\tt 0}}

\newcommand{\sem}[1]{[\![ #1 ]\!]}

\newcommand{\eqdef}{\stackrel{\it def}{=}}

\renewcommand{\emptyset}{\varnothing}

\newcommand{\var}[1]{{\it var}(#1)}

\newcommand{\mut}[1]{\mathbb{\Lparen} \, #1 \, \mathbb{\Rparen}}
\newcommand{\elem}[1]{#1}
\newcommand{\mutpair}[2]{\ilam{#2, #1}{}}

\begin{document}

\lstset{
language=Java,
basicstyle=\ttfamily\small, 
numbers=left, 
numberstyle=\ttfamily\small,
stepnumber=2,
numbersep=2pt}

\title{Deadlock detection in linear recursive programs}

\author{Elena Giachino\and Cosimo Laneve}
\institute{Dept.~of Computer Science and Egineering, Universit\`a di Bologna -- 
           INRIA FOCUS
           {\tt $\{$giachino,laneve$\}$@cs.unibo.it} }

\maketitle

\pagestyle{plain}

\begin{abstract}
Deadlock detection in recursive programs that admit dynamic resource 
creation is extremely complex and solutions either give imprecise answers or do not
scale.

We define an algorithm for detecting deadlocks of \emph{linear recursive programs} of
a basic model. 
The theory that underpins the algorithm is a generalization of the theory of 
permutations of names to so-called \emph{mutations}, which transform tuples by
introducing duplicates and fresh names.

Our algorithm realizes the back-end of
deadlock analyzers for object-oriented programming languages,
once the association programs/basic-model-programs has been defined as front-end.
\end{abstract}

\section{Introduction}
\label{sec.introduction}

Deadlocks in concurrent programs are detected by 
building graphs of dependencies 
$(x,y)$ between resources, meaning that the release of a resource referenced by 
$x$ depends on the release of the resource referenced by $y$.
The absence of cycles in the graphs entails deadlock freedom. 
When programs have infinite states, 
the deadlock detection tools use finite models that are excerpted
from the dependency graphs to ensure termination.

The most powerful deadlock analyzer we are aware of is {\sc TyPiCal},
a tool developed for pi-calculus by 
Kobayashi~\cite{Typicaltool,Kobayashi1998,Kobayashi2004,Kobayashi06}.
This tool uses a clever technique for deriving 
inter-channel dependency information and
is able to deal with several recursive behaviors and the creation of new 
channels without
using any pre-defined order of channel names.
Nevertheless, since 
{\sc TyPiCal} is based on an inference system, there are recursive behaviors
that escape its accuracy.
For instance, it returns false positives when recursion is mixed up 
with \emph{delegation}. To illustrate the issue we consider the following deadlock-free pi-calculus factorial program
\begin{lstlisting}[numbers=none]
*factorial?(n,(r,s)).
  if n=0 then r?m. s!m else new t in 
                       (r?m. t!(m*n)) | factorial!(n-1,(t,s)) 
\end{lstlisting}
In this code, {\tt factorial} returns the value (on the channel {\tt s}) by 
\emph{delegating}
this task to the recursive invocation, if any. In particular,
the initial invocation of {\tt factorial}, which is {\tt r!1 | factorial!(n,(r,s))}, 
performs a synchronization between {\tt r!1} and the input {\tt r?m} in 
the continuation of {\tt factorial?(n,(r,s))}. In turn, this may delegate 
the computation of the factorial to a subsequent synchronization on a 
new channel {\tt t}.
{\sc TyPiCal} signals a deadlock on the two inputs {\tt r?m} because it fails in connecting
the output {\tt t!(m*n)} with them. 

\smallskip

\emph{The technique we develop in this paper allows us to demonstrate the 
deadlock freedom of programs 
like the one above}.

\smallskip

To ease program reasoning, our technique relies on an abstraction process that 
extracts the dependency constraints in programs
\begin{itemize}
\item
by dropping primitive data types and values;
\item
by highlighting dependencies between pi-calculus actions; 
\item
by overapproximating statement behaviors, namely 
collecting the dependencies and the invocations in the two branches of the
conditional (the set union operation is modeled by
    $\sparop$).
\end{itemize}
This abstraction process is currently performed by a formal inference system 
that does not target pi-calculus, but it is defined for a {\sc Java}-like
programming language, called {\tt ABS}~\cite{ABS}, see 
Section~\ref{sec.assessments}. 
Here, pi-calculus has been considered for expository 
purposes. The {\tt ABS} program corresponding to the pi-calculus {\tt factorial} may 
be downloaded from~\cite{DAT}; readers that are familiar with {\sc Java} may find
the code in the Appendix~\ref{sec.Java}.
As a consequence of the abstraction operation we get the function

\smallskip

$\qquad{\tt factorial}(r,s) = (r,s) \sparop (r,t) \sparop {\tt factorial}(t,s)$

\smallskip

\noindent
where $(r,s)$ shows the dependency between the actions {\tt r?m} and {\tt s!m} and
$(r,t)$ the one between {\tt r?m} and {\tt t!(m*n)}.
The semantics of the abstract ${\tt factorial}$ is  defined operationally by unfolding the recursive
invocations. In particular, the unfolding of ${\tt factorial}(r,s)$ yields the
sequence of abstract states (free names in the definition of ${\tt factorial}$ are replaced
by fresh names in the unfoldings)

\smallskip

\noindent
$\begin{array}{@{\!}r@{\,}l}
 {\tt factorial}(r,s) \lred{} & (r,s) \sparop (r,t) \sparop {\tt factorial}(t,s) 
 \\
 \lred{} & (r,s) \sparop (r,t) \sparop (t,s) \sparop (t,u) \sparop {\tt factorial}(u,s) 
 \\
 \lred{} & (r,s) \sparop (r,t) \sparop (t,s) \sparop (t,u) \sparop 
 (u,s) \sparop (u,v) 
 \\
 &  \sparop {\tt factorial}(v,s)
 \\
 \lred{} & \quad \cdots
\end{array}
$

\smallskip

We demonstrate that 
the abstract {\tt factorial} (and, therefore, the
foregoing pi-calculus code) never manifests a 
circularity by using a \emph{model checking} technique. This despite the fact that
the model of {\tt factorial} has infinite states. 
In particular, we are able to
decide the deadlock freedom by analyzing finitely many states -- precisely
three -- of  {\tt factorial}.

\paragraph{Our solution.}
We introduce a basic recursive model, 
called \emph{lam programs} -- lam is an acronym for \emph{deadLock Analysis Model} 
-- that are collections of function definitions and a main term to evaluate.
For example, 

\smallskip

$\begin{array}{rl}
\bigl( & {\tt factorial}(r,s) = (r,s) \sparop (r,t) \sparop {\tt factorial}(t,s) \; , 
{\tt factorial}(r,s) \; \; \bigr)
\end{array}$

\smallskip

\noindent
defines ${\tt factorial}$ and the main term ${\tt factorial}(r,s)$.
Because lam programs feature recursion and dynamic name creation -- \emph{e.g.}~the 
free name $t$ in the definition of 
${\tt factorial}$ -- the  model is
not finite state (see Section~\ref{sec.fulllanguage}).

In this work we address the %
\begin{question}\label{question}
Is it decidable whether the computations of a lam 
program will ever produce a circularity?
\end{question}
and the main contribution is the positive answer when programs are \emph{linear 
recursive}. 

To begin the description of our solution, we notice that, if lam programs are
non-recursive then detecting circularities is as simple as unfolding the invocations 
in the main term. 
In general, as in case of ${\tt factorial}$, the unfolding may not terminate.
Nevertheless, the following two conditions may ease our answer:
\begin{itemize}
\item[(i)] the functions in the program are \emph{linear recursive},
that is (mutual) recursions have at most one recursive invocation -- such as
${\tt factorial}$;
\item[(ii)] function invocations do not show duplicate arguments 
and function definitions do not have free names.
\end{itemize}
When (i) and (ii) hold, as in the program

\smallskip

\qquad \qquad $\bigl(\M(x,y,z) = (x,y) \sparop \M(y,z,x), \; \M(u,v,w) \bigr)\;$ ,

\smallskip
\noindent
recursive functions may be considered as \emph{permutations of names}  -- technically 
we define a notion of \emph{associated (per)mutation} --  and
the corresponding theory~\cite{Comtet} guarantees that, by repeatedly  
applying  a same permutation 
to a tuple of names, at some point, one obtains the initial tuple. 
This point, which is known as the \emph{order} of the permutation, allows one to 
define the following algorithm for Question~\ref{question}:
\begin{enumerate}
\item
compute the order of the permutation associated to the 
function in the lam and 
\item
correspondingly unfold the term
to evaluate.
\end{enumerate}
For example, the permutation of $\M$ has order 3. Therefore,
it is possible to stop the evaluation of $\M$ after the third unfolding (at the state
$(u,v) \sparop (v,w) \sparop (w,u) \sparop \M(u,v,w)$) 
because every dependency pair produced afterwards will belong to the relation $(u,v) \sparop (v,w) \sparop (w,u)$.

When the constraint (ii) is dropped, as in ${\tt factorial}$,
the answer to Question~\ref{question} 
is not simple anymore. However, the above analogy with permutations has been a 
source of inspiration for us. 
\begin{figure*}
 \centering
\includegraphics[width=1.2\textwidth]{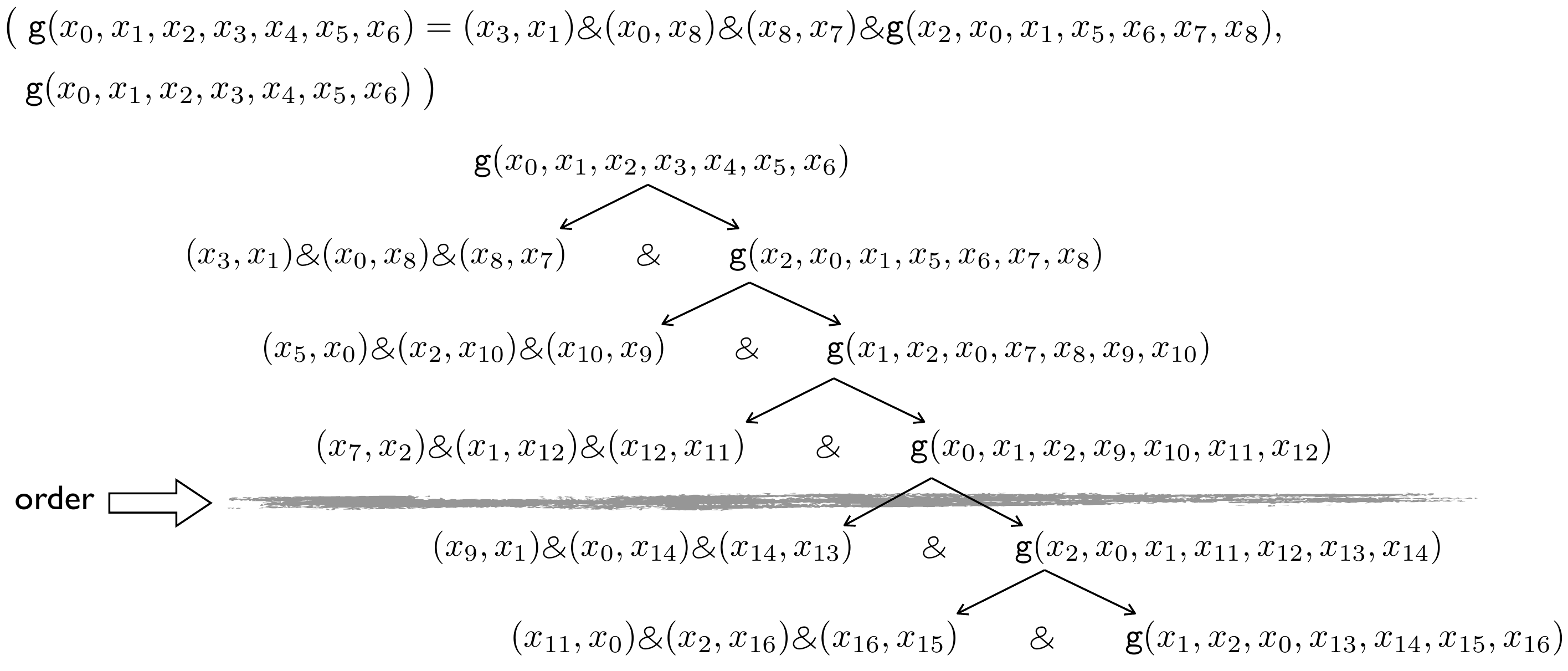}
\caption{\label{fig.lamunfolding}A lam program and its unfolding}
\end{figure*}

Consider the main term ${\tt factorial}(r,s)$. 
Its evaluation will never display 
${\tt factorial}(r,s)$ twice, as well as any other invocation
in the states, because the first argument of the recursive invocation 
is free. Nevertheless, we notice that, 
from the second state -- namely $(r,s) \sparop (r,t) \sparop {\tt factorial}(t,s)$ --  onwards, the invocations of ${\tt factorial}$
are not identical, but \emph{may be identified by a map} that
\begin{itemize}
\item[--]
associates names created in the last evaluation step to past names,
\item[--]
is the identity on other names.
\end{itemize}
The definition of this map, called \emph{flashback}, requires that the transformation associated to a lam function, called 
\emph{mutation}, also records the name creation. 
In fact, the theory of mutations allows us to map ${\tt factorial}(t,s)$ 
back to ${\tt factorial}(r,s)$ by recording 
that $t$ has been created after $r$, \emph{e.g.}~ $r \lessc t$.

We generalize the result about permutation orders (Section~\ref{sec.mutationsandflashbacks}):
\begin{quote}
\emph{by repeatedly applying a same mutation 
to a tuple of names, %
at some point we obtain a
tuple that is identical, up-to a flashback, to a tuple in the past.}
\end{quote}
As for permutations, this point is
the \emph{order} of the mutation, which (we prove) it is possible to compute in similar ways.

However, unfolding a function  
as many times as the order of the associated mutation may not be sufficient
for displaying circularities. 
This is unsurprising because the arguments about mutations and flashbacks
focus on function invocations and do not account for dependencies. In the 
case of lams where (i) and (ii) hold, these arguments were sufficient because permutations 
reproduce \emph{the same} dependencies of past invocations. In the case of
mutations, this is not true anymore as displayed by the function  $\N$ in
Figure~\ref{fig.lamunfolding}. This function
has order 3 and the first three unfoldings of 
$\N(x_0, x_1, x_2, x_3, x_4, x_5, x_6)$  are those above the horizontal line. 
While there is a flashback from 
$\N(x_0, x_1, x_2, x_9, x_{10}, x_{11}, x_{12})$ to $\N(x_0, x_1, x_2, x_3, x_4, x_5, x_6)$, the pairs produced up-to the third unfolding

\smallskip

$\begin{array}{l}
(x_3,x_1) \sparop (x_0,x_8) \sparop (x_8,x_7) 
\sparop 
(x_5,x_0) \sparop (x_2,x_{10}) \sparop (x_{10},x_9) 
\\ 
\sparop 
(x_7,x_2) \sparop (x_1,x_{12}) \sparop (x_{12},x_{11})
\end{array}$

\smallskip
\noindent
do not manifest any circularity.  Yet, two additional unfoldings
(displayed below the horizontal line of Figure~\ref{fig.lamunfolding}),
show the circularity

\smallskip

$\begin{array}{c}
  (x_0,x_8) \sparop (x_8,x_7) \sparop (x_7,x_2) \sparop (x_2, x_{10})  \sparop (x_{10}, x_9)
\\
  \sparop (x_9,x_1) \sparop (x_1,x_{12})\sparop (x_{12},x_{11}) \sparop (x_{11}, x_0) \; .
\end{array}$

\smallskip

\begin{figure}
 \centering
\includegraphics[width=0.6\textwidth]{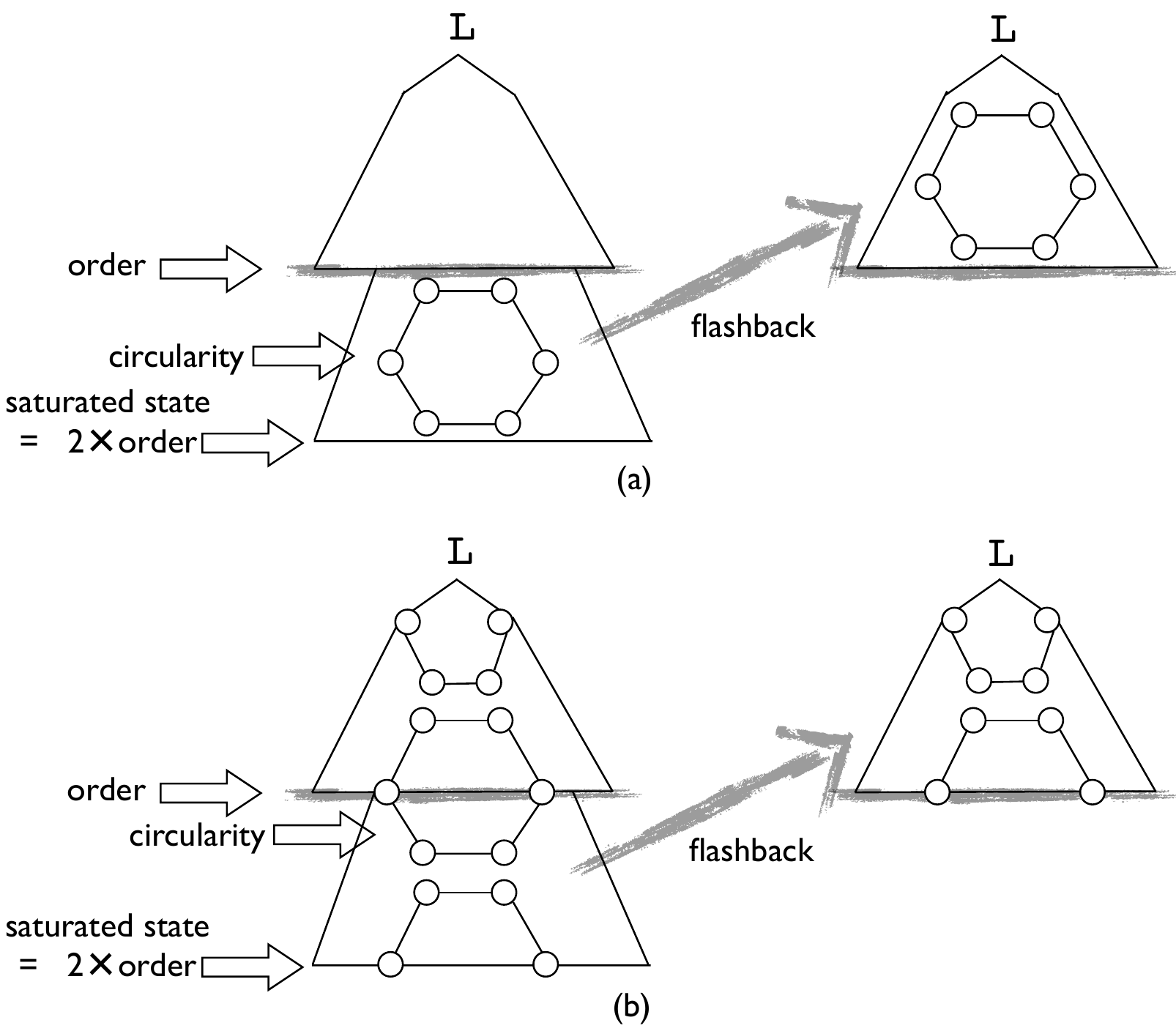}
\caption{\label{fig.flashbackcirc}Flashbacks of circularities}
\vspace{-.5cm}
\end{figure}
In Section~\ref{sec.algorithm} we prove that a sufficient condition 
for deciding whether a lam program as in Figure~\ref{fig.lamunfolding} will ever produce 
a circularity is to unfold the function $\N$
\emph{up-to two times} the order of the associated mutation -- this state will be 
called \emph{saturated}. If no circularity is manifested in the saturated state then 
the lam is ``circularity-free''. 
This supplement of evaluation is due to the existence of two alternative 
ways for
creating circularities. A first way is when the circularity is given by
the dependencies produced by the unfoldings from the order to the saturated state. Then,
our theory guarantees that the circularity is also present in the unfolding of 
$\N$ till the order -- see Figure~\ref{fig.flashbackcirc}.a. A second way is when 
the dependencies of the circularity are produced by (1) the unfolding till the order 
\emph{and} by (2) the unfolding from the order till the saturated state -- these are the 
so-called \emph{crossover circularities} -- see 
Figure~\ref{fig.flashbackcirc}.b. Our theory allows us to map dependencies
of the evaluation (2)  to those of
the evaluation (1) and the flashback may break the circularity -- in this case,
the evaluation till the saturated state is necessary to collect enough informations.
Other ways for creating circularities are excluded. The intuition behind this fact is that the behavior of the function (the dependencies) repeats itself following the same pattern every order-wise unfolding. Thus it is not possible to reproduce a circularity that crosses more than one order without having already a shorter one.
The algorithm for detecting circularities in linear recursive lam programs is
detailed in Section~\ref{sec.thealgorithm-vero}, together with a discussion about its
computational cost.

We have prototyped our algorithm~\cite{DAT}.
In particular, the prototype (1) uses a (standard but not straightforward) 
\emph{inference system} that we developed for deriving behavioral types with 
dependency informations out of {\tt ABS} programs~\cite{GL2013a} and
(2) has an add-on translationg these behavioral types into lams. We have been
 able to verify an 
industrial case study 
developed by SDL Fredhoppper -- more than 2600 lines of code -- in 31 seconds.
Details about our prototype and a comparison 
with other deadlock analysis tools can be found in Section~\ref{sec.assessments}.
There is no space in this contribution to discuss the inference system: the interested
readers are referred to~\cite{GL2013a}.

\section{Generalizing permutations: mutations and flashbacks}
\label{sec.mutationsandflashbacks}

Natural numbers are ranged over by $a$, $b$, $i$, $j$, $m$, $n$, $\dots$, possibly indexed.
Let ${\sf V}$ be an infinite set of names, ranged over by $x, y, z, \cdots$.
We will use partial order relations  on names
-- relations that are reflexive, antisymmetric,
and transitive --, ranged over by
$\por{V}, \por{V}', \por{I}, \cdots$.
Let $x \in \por{V}$ if, for some $y$, either 
 $(x,y) \in \por{V}$ or $(y,x) \in \por{V}$. Let also $\var{\por{V}} = 
 \{ x \; | \; x \in \por{V} \}$.
For notational convenience, we write $\wt{x}$ 
when we refer to a list of names $x_1,\dots,x_n$.

Let $\por{V} \transclosure \wt{x} \lessc \wt{z}$, 
with $\wt{x} \in \por{V}$ and $\wt{z} \notin \por{V}$,
be the least partial order containing the set 
$\por{V} \cup \{ (y,z) \; | \; x \in \wt{x} \; \mbox{and} \; (x,y) \in \por{V} 
\; \mbox{and} \; z\in \wt{z} \}$.
That is, $\wt{z}$ become \emph{maximal names} in
$\mathbb{V} \transclosure \wt{x} \lessc 
\wt{z}$.
For example, 
\begin{itemize}
\item[--]
$\{ (x,x) \} \transclosure x \lessc z = \{ (x,x), (x,z) , (z,z) \}$;
\item[--]
if $\por{V} = \{ (x,y), (x',y') \}$ (the reflexive pairs
are omitted) then 
$\por{V} \transclosure y\lessc z$ is the reflexive and transitive closure of 
$\{
(x,y), (x',y'),$ $(y,z) \}$;

\item[--]
if $\por{V} = \{ (x,y), (x,y') \}$ (the reflexive pairs
are omitted) then 
$\por{V} \transclosure x \lessc 
z$ is the reflexive and transitive closure of 
$\{
(x,y), (x,y'),$ $(y,z), (y',z)
\}$.
\end{itemize}
Let $x \leq y \in \por{V}$ be $(x,y) \in \por{V}$. 

\begin{definition}
A \emph{mutation} of a tuple of names, denoted $\mut{\elem{a}_1, \cdots , \elem{a}_n}$ 
where $1 \leq \elem{a}_1, \cdots, a_n \leq 2\times n$,
transforms a pair
$\mutpair{(x_1, \cdots,  x_n)}{\por{V}}$ into a pair $\mutpair{(x_1', \cdots, x_n')}{\por{V}'}$ 
as follows. Let $\{b_{1}, \cdots , b_{k}\} = \{a_1, \cdots , a_n\} \setminus
\{1,2, \cdots , n\}$ and let $z_{b_{1}}, \cdots , z_{b_{k}}$ be $k$ pairwise 
different fresh names. [That is names not occurring either in $x_1, \cdots , x_n$ 
or in $\por{V}$.] Then
\begin{itemize}
\item[--]
if $1 \leq \elem{a}_i \leq n$ then $x_i' = x_{\elem{a}_i}$;
\item[--]
if $\elem{a}_i > n$ then
$x_i'= z_{a_i}$;
\item[--]
$\por{V}' = \por{V} \transclosure x_1, \cdots, x_n \lessc z_{{i_1}}, \cdots , z_{{i_k}}$.
\end{itemize}

The mutation $\mut{\elem{a}_1, \cdots , \elem{a}_n}$ of 
$\mutpair{(x_1, \cdots, x_n)}{\por{V}}$ 
will be written $\mutpair{(x_1, \cdots, x_n)}{\por{V}}$ $\lred{\mut{\elem{a}_1, \cdots , 
\elem{a}_n}} \mutpair{(x_1', \cdots, x_n')}{\por{V}'}$ and the label $\mut{\elem{a}_1, \cdots , \elem{a}_n}$ is omitted when 
the mutation is clear from the context.
Given a mutation $\mu=\mut{\elem{a}_1, \cdots , \elem{a}_n}$, we define the application of $\mu$ to an index $i$, $1\leq i\leq n$, as $\mu(i)=\elem{a}_i$.
\end{definition}

Permutations are mutations $\mut{\elem{a}_1, \cdots , \elem{a}_n}$ where the elements
are pairwise different and belong 
to the set $\{1, 2, \cdots,  n\}$ (e.g. $\mut{2,3,5,4,1}$).
In this case the partial order $\por{V}$ never changes and therefore it is useless.
Actually, our terminology and statements below are inspired by the corresponding 
ones for permutations.
A mutation differs from a permutation because it can exhibit repeated elements, or even new elements (identified by $n+1 \leq \elem{a}_i 
\leq 2\times n$, for some $\elem{a}_i$). 
For example, by successively applying the mutation $\mut{2,3,6,1,1}$
to $\mutpair{(x_1, x_2, x_3, x_4, x_5)}{\por{V}}$, with $\por{V} = 
\{ (x_1,x_1), \cdots , (x_5,x_5)\}$ and $\wt{x} = x_1, x_2, x_3, x_4, x_5$, we obtain

\smallskip

$\begin{array}{r@{\quad}l}
\mutpair{(x_1, x_2, x_3, x_4, x_5)}{\por{V}}
 \quad \lred{} & \mutpair{(x_2, x_3, y_1, x_1, x_1)}{\por{V}_1}
\\
\lred{} & \mutpair{(x_3, y_1, y_2, x_2, x_2)}{\por{V}_2}
\\
\lred{} & \mutpair{(y_1, y_2, y_3, x_3, x_3)}{\por{V}_3}
\\
\lred{}& \mutpair{(y_2, y_3, y_4, y_1, y_1)}{\por{V}_4}
\\
\lred{}& \cdots
\end{array}$

\smallskip
\noindent
where $\por{V}_1 = \por{V} \transclosure \wt{x} \lessc y_1$ and, for $i\geq 1$,
$\por{V}_{i+1} = \por{V}_{i} \transclosure y_{i} \lessc y_{i+1}$.
In this example, $6$ identifies a new name to be added at each application of the 
mutation. The new name created at each step is a maximal one for the partial order.

We observe that, by definition, $\mut{2,3,6,1,1}$ and $\mut{2,3,7,1,1}$
define a same transformation of names. That is, the choice of the natural between 6 and 10 
is irrelevant in the definition of the mutation. Similarly for the mutations
$\mut{2,3,6,1,6}$ and $\mut{2,3,7,1,7}$.

\begin{definition}
\label{def.eqmutations}
Let $\mut{a_1, \cdots , a_n} \approx \mut{a_1' , \cdots , a_n'}$ if
there exists a bijective function $f$ from $[n+1..2\times n]$ to $[n+1..2\times n]$
such that:
\begin{enumerate}
\item
$1 \leq a_i \leq n$ implies $a_i' = a_i$;
\item
$n+1 \leq a_i \leq 2 \times n$ implies $a_i' = f(a_i)$.
\end{enumerate}
\end{definition}
We notice that $\mut{2,3,6,1,1} \approx \mut{2,3,7,1,1}$ and
$\mut{2,3,6,1,6} \approx \mut{2,3,7,1,7}$. However 
$\mut{2,3,6,1,6} \not \approx \mut{2,3,6,1,7}$; in fact these two mutations define
different transformations of names.

\begin{definition}
Given a partial order $\por{V}$, a $\por{V}$-\emph{flashback} is an injective renaming 
$\rho$ on names such that 
$\rho(x) \leq x \in \por{V}$.
\end{definition}
In the above sequence of mutations of $(x_1, x_2, x_3, x_4, x_5)$
there is a $\por{V}_4$-flashback from  $(y_2, y_3, y_4, y_1, y_1)$ to 
$(x_2, x_3, y_1, x_1, x_1)$.
In the following, flashbacks will be also applied to tuples: 
$\rho(x_1, \cdots,$ $x_n) \eqdef (\rho(x_1), \cdots, \rho(x_n))$.

In case of mutations that are permutations, a flashback
is the identity renaming  and the following statement is folklore. 
Let $\mu$ be a mutation. We write  $\mu^m$ for the application of $\mu$
$m$ times, namely $\mutpair{(x_1, \cdots , x_n)}{\por{V}} \lred{\mu^m}
\mutpair{(y_1, \cdots , y_n)}{\por{V}'}$ abbreviates
$\underbrace{\mutpair{(x_1, \cdots , x_n)}{\por{V}} \lred{\mu} \cdots \lred{\mu} \mutpair{(y_1, \cdots , y_n)}{\por{V}'}}_{m \;
{\rm times}}$. 

\begin{proposition}
Let $\mu = \mut{\elem{a}_1, \cdots , \elem{a}_n}$ and

\smallskip

$\begin{array}{rl}
\mutpair{(x_1, \cdots, x_n)}{\por{V}} \lred{\mu} 
& \mutpair{(x_1', \cdots, x_n')}{\por{V}'} 
\\
\lred{\mu^m} & \mutpair{(y_1, \cdots, y_n)}{\por{V}''}  
\\
\lred{\mu} & \mutpair{(y_1', \cdots, y_n')}{\por{V}'''}
\end{array}$

\smallskip

If there is a $\por{V}''$-flashback $\rho$ such that  $\rho(y_1, \cdots, y_n) = 
(x_1, \cdots, x_n)$ then there is a 
$\por{V}'''$-flashback from $(y_1', \cdots, y_n')$ to $(x_1', \cdots, x_n')$.
\end{proposition}

\begin{proof}
Let $\rho'$ be the relation 
$y_i' \mapsto x_i'$, for every $i$. Then

1)~$\rho'$ is a mapping: $y_i' = y_j'$ implies $x_i' = x_j'$. In 
fact, $y_i' = y_j'$ means that either (i) $1 \leq \elem{a}_i, \elem{a}_j \leq n$ or (ii) $\elem{a}_i, \elem{a}_j > n$.
In subcase (i) $y_{\elem{a}_i} = y_{\elem{a}_j}$, by definition of
mutation. Therefore $\rho(y_{\elem{a}_i}) = \rho( y_{\elem{a}_j})$ that in turn
implies  $x_{\elem{a}_i} = x_{\elem{a}_j}$. From this last equality we obtain
$x_i' = x_j'$. In subcase (ii), $\elem{a}_i = \elem{a}_j$ and the implication follows
by the fact that $\mut{\elem{a}_1, \cdots , \elem{a}_n}$ is a mutation.

2)~$\rho'$ is injective: $x_i' = x_j'$ implies $y_i' = y_j'$. If
$x_i' \in \{ x_1, \cdots , x_n \}$ then $1 \leq \elem{a}_i, \elem{a}_j 
\leq  n$.
Therefore, by the definition of mutation, $x_{\elem{a}_i} = x_{\elem{a}_j}$ and,
because $\rho$ is a flashback, $y_{\elem{a}_i} = y_{\elem{a}_j}$. By this last 
equation $y_i' = y_j'$.
If $x_i' \notin \{ x_1, \cdots , x_n \}$ then $\elem{a}_i >  n$ and $\elem{a}_i = \elem{a}_j$. Therefore $y_i' = y_j'$ by definition
of mutation.

3)~$\rho'$ is a flashback: $x_i' \neq y_i'$ implies $ x_i' \leq y_i' \in \por{V}'''$.
If $1 \leq \elem{a}_i \leq n$ then $y_i' = y_{\elem{a}_i}$ and $x_i' = x_{\elem{a}_i}$.
Therefore $y_{\elem{a}_i} \neq x_{\elem{a}_i}$ and we conclude by the hypothesis about $\rho$ that $\rho'(y_{\elem{a}_i})$ satisfies the constraint in the
definition of flashback.
If $\elem{a}_i > n$ then $x_1, \cdots , x_n \leq x_i' \in \por{V}'$.
Since $\rho(y_i) = x_i$, by the hypothesis about $\rho$,  $x_i \leq y_i  \in 
\por{V}''$. Therefore, by definition of mutation, $x_i' \leq y_i \in  
\por{V}''$. We derive $x_i' \leq y_i' \in  
\por{V}'''$ by transitivity because $\por{V}'' \subseteq \por{V}'''$ and $y_i 
\leq y_i' \in \por{V}'''$.
\end{proof}

The following 
 Theorem~\ref{thm.mainthm}
  generalizes the property that every permutation has an
\emph{order}, which is
the number of applications that return the initial tuple. 
In the theory of permutations, the order 
is the least common multiple, in short $\lcm$, of the lengths of the cycles 
of the permutation. %
This result is clearly false for mutations because of the
presence of duplications and of fresh names. The generalization that holds in our
setting uses flashbacks instead of identities.
We begin by extending the notion of cycle.

\begin{definition}[Cycles and sinks] 
Let $\mu=\mut{\elem{a}_1,\cdots, \elem{a}_n}$ be a mutation
and let $1 \leq \elem{a}_{i_1}, \, \cdots, \, \elem{a}_{i_\ell} \leq n$ be pairwise different naturals. Then:
\begin{enumerate}
\item[i.] the term $(\elem{a}_{i_1} \, \cdots \, \elem{a}_{i_\ell})$ is a \emph{cycle}
of $\mu$ whenever $\mu(\elem{a}_{i_j})=\elem{a}_{i_{j+1}}$, with $1\leq j \leq \ell-1$, and  $\mu(\elem{a}_{i_\ell})=\elem{a}_{i_1}$ (i.e., $(\elem{a}_{i_1} \, \cdots \, \elem{a}_{i_\ell})$ is the ordinary permutation cycle);

\item[ii.] the term $[\elem{a}_{i_1} \, \cdots \, \elem{a}_{i_{\ell-1}}]_{\elem{a}_{i_{\ell}}}$ is a \emph{bound sink} of $\mu$ whenever $a_{i_1} \notin 
\{ \elem{a}_1,\cdots, \elem{a}_n \}$,
$\mu(\elem{a}_{i_j})=\elem{a}_{i_{j+1}}$, with $1\leq j \leq \ell-1$, and 
$\elem{a}_{i_{\ell}}$ belongs to a cycle;

\item[iii.] the term $[\elem{a}_{i_1} \, \cdots \, \elem{a}_{i_{\ell}}]_{\elem{a}}$, with $n< a \leq 2 \times n$,  is a
\emph{free sink} of $\mu$ 
whenever $a_{i_1} \notin 
\{ \elem{a}_1,\cdots, \elem{a}_n \}$ and
$\mu(\elem{a}_{i_j})=\elem{a}_{i_{j+1}}$, with $1\leq j \leq \ell-1$ and
$\mu(\elem{a}_{i_\ell})= a$.
\end{enumerate}
The \emph{length of a cycle} is the number of elements in the cycle; the 
\emph{length of
a sink} is the number of the elements in the square brackets.
\end{definition}
For example the mutation $\mut{5, 4, 8, 8, 3, 5, 8, 3, 3}$ has cycle $(3,8)$ and
has bound sinks $[1,5]_3$, $[6,5]_3$, $[9]_3$, $[2,4]_8$, and $[7]_8$.
 The mutation $\mut{6,3,1,8,7,1,8}$ has cycle $(1,6)$, has bound sink
 $[2,3]_1$ and free sinks $[4]_8$ and $[5,7]_{8}$. 

Cycles and sinks are an alternative description of a mutation. 
For instance
$(3,8)$ means that the mutation moves the element in position $8$  to the element 
in position $3$
and the one in position $3$ to the position $8$; the free sink $[5,7]_{8}$
means that the element in position $7$ goes to the position $5$, whilst a 
fresh name goes in position $7$.

\begin{theorem}
\label{thm.mainthm}
Let $\mu$ be a mutation, $\ell$ be the {\lcm} of the length  of its cycles, 
$\ell'$ and $\ell''$ be the lengths of its longest bound sink and free
sink, respectively. Let also 
$k \eqdef {\tt max}\{ \ell+\ell', \; \ell'' \}$. 
Then there exists $0 \leq h < k$ such that 
$\mutpair{(x_1, \cdots, x_n)}{\por{V}} \lred{\mu^{h}} 
\mutpair{(y_1, \cdots, y_n)}{\por{V}'}$ $\lred{\mu^{k-h}} 
\mutpair{(z_1, \cdots, z_n)}{\por{V}''}
$ and $\rho(z_1, \cdots, z_n) = (y_1, \cdots, y_n)$, for some
$\por{V}''$-flashback $\rho$.
The value $k$ is called \emph{order of $\mu$} and denoted by
$\ordermut{\mu}$.
\end{theorem}

\begin{proof}
Let $\mu=\mut{\elem{a}_1,\cdots,\elem{a}_n}$ be a mutation, and let
$A=\{1,2,\dots,n\}\setminus \{\elem{a}_1,\cdots,\elem{a}_n\}.$ 

If $A=\varnothing$, then $\mu$ is a permutation; hence, by the theory of
permutations, the theorem is immediately proved taking $\rho$ as the identity 
and $h=0$.

If $A\neq\varnothing$ then let $a\in A$. By definition, $a$ must be the
first element of (i) a  bound sink or (ii) a free sink of $\mu$.
We write either
$a\in A_{(i)}$ or $a\in A_{(ii)}$ if $a$ is the first element of
a bound or free sink, respectively.

In subcase (i), let $\ell_a'$ be the length  of the bound sink with subscript $a'$
and $\ell_{a'}$ be the 
length of the cycle of $a'$. 
We observe that in
$\mutpair{(x_1, \cdots, x_n)}{\por{V}} \lred{\mu^{\ell_a'}} 
\mutpair{(x_1', \cdots, x_n')}{\por{U}}$ $\lred{\mu^{\ell_{a'}}} 
\mutpair{(x_1'', \cdots, x_n'')}{\por{W}}
$ we have $x_{a'}' = x_{a'}''$.

In subcase (ii), let $\ell_a''$ be the length of the free sink.  
We observe that in
$\mutpair{(x_1, \cdots, x_n)}{\por{V}} \lred{\mu^{\ell_a''}} 
\mutpair{(x_1', \cdots, x_n')}{\por{U}}$ we have $x_{a} \leq x_{a}' \in 
\por{U}$, by definition of mutation.

Let $\ell$, $\ell'$ and $\ell''$ as defined in the theorem. Then, if
$\ell+\ell' \geq \ell''$ we have that
$\mutpair{(x_1, \cdots, x_n)}{\por{V}} \lred{\mu^{\ell'}} 
\mutpair{(y_1, \cdots, y_n)}{\por{V}'}$ $\lred{\mu^{\ell}} 
\mutpair{(z_1, \cdots, z_n)}{\por{V}''}$ and 
$\rho(z_1, \cdots, z_n) = (y_1, \cdots, y_n)$, where $\rho = 
[z_1 \mapsto y_1, \cdots , z_n \mapsto y_n]$ is a 
$\por{V}''$-flashback.
If $\ell + \ell' < \ell''$ then
$\mutpair{(x_1, \cdots, x_n)}{\por{V}} \lred{\mu^{\ell''-\ell}} 
\mutpair{(y_1, \cdots, y_n)}{\por{V}'}$ $\lred{\mu^{\ell}} 
\mutpair{(z_1, \cdots, z_n)}{\por{V}''}$ and 
$\rho(z_1, \cdots, z_n) = (y_1, \cdots, y_n)$, where $\rho = 
[z_1 \mapsto y_1, \cdots , z_n \mapsto y_n]$ is a 
$\por{V}''$-flashback.
\end{proof}

For example, $\mu=\mut{6,3,1,8,7,1,8}$, has a cycle $(1,6)$, bound sink
$[2,3]_1$ and free sinks $[4]_8$ and $[5, 7]_8$. Therefore
$\ell=2$, $\ell'=2$ and $\ell'' =2$. In this case, the values $k$ and $h$ of
Theorem~\ref{thm.mainthm} are 4 and 2, respectively. In fact,
if we apply the mutation $\mu$ four times to the pair $\mutpair{(x_1,x_2,x_3,x_4,x_5,x_6,x_7)}{\por{V}}$, where $\por{V}=\{(x_i,x_i)\,|\,1\leq i \leq 7\}$ we  obtain

\smallskip

$\begin{array}{ll}
\mutpair{(x_1,x_2,x_3,x_4,x_5,x_6,x_7)}{\por{V}} 
& \lred{\mu} \; \mutpair{(x_6,x_3,x_1,y_1,x_7,x_1,y_1)}{\por{V}_1}\\
& \lred{\mu} \; \mutpair{(x_1,x_1,x_6,y_2,y_1,x_6,y_2)}{\por{V}_2}\\
& \lred{\mu} \; \mutpair{(x_6,x_6,x_1,y_3,y_2,x_1,y_3)}{\por{V}_3}\\
& \lred{\mu} \;  \mutpair{(x_1,x_1,x_6,y_4,y_3,x_6,y_4)}{\por{V}_4}
\end{array}$

\smallskip
\noindent
where $\por{V}_1 = \por{V} \transclosure x_1,x_2,x_3,x_4,x_5,x_6,x_7\lessc y_1$
and, for $i \geq 1$,  $\por{V}_{i+1} = \por{V}_i \transclosure y_{i-1} \lessc y_i$.
We notice that there is a $\por{V}_4$-flashback $\rho$ from
$(x_1,x_1,x_6,y_4,y_3,x_6,y_4)$ (produced by $\mu^4$) to 
$(x_1,x_1,x_6,y_2,y_1,$ $x_6,y_2)$ (produced by $\mu^2$).

\section{The language of lams}
\label{sec.fulllanguage}

We use an infinite set of
\emph{function names}, ranged over $\M$, $\M'$, $\N$, $\N'$,$\ldots$, which is
disjoint from the set ${\sf V}$ of Section~\ref{sec.mutationsandflashbacks}. 
A \emph{lam program} is a tuple $\bigl(\M_1(\wt{x_1}) = \P_{1}, \cdots , 
\M_\ell(\wt{x_\ell}) = \P_{\ell}, \P \bigr)$ where $\M_i(\wt{x_i}) = \P_{i}$ 
are \emph{function definitions} and $\P$ is the
\emph{main lam}. The syntax of $\P_{i}$ and $\P$ is

\smallskip

$\begin{array}{rl}
\P \quad ::= \qquad & \pinull \quad | \quad (x,y) 
 \quad | \quad \M(\wt{x}) \quad | \quad \P 
\sparop \P  \quad | \quad \P \seq \P
\end{array}$

\smallskip

Whenever parentheses are omitted, the operation ``$\sparop$'' has precedence
over ``$\seq$''. We will shorten $\P_1 \sparop \cdots \sparop \P_n$ into 
$\prod_{i \in 1..n}\P_i$.
Moreover, we use $\T$ to range over lams that do not contain function invocations.

Let $\var{\P}$ 
be the set of names in $\P$.
In a function definition 
$\M(\wt{x}) = \P$, $\wt{x}$ are the \emph{formal parameters} and
 the occurrences of names $x \in \wt{x}$ 
in $\P$ are \emph{bound}; the names $\var{\P} \setminus 
\wt{x}$ are \emph{free}.

In the syntax of $\P$, the operations ``$\sparop$'' and ``$\seq$'' are associative, commutative with $\pinull$ being the identity. 
Additionally the following axioms hold ($\T$ does not contain function invocations)

\smallskip

$\begin{array}{c}
\T \sparop \T = \T 
\qquad 
\T \seq \T =  \T
\qquad
\T \sparop (\P' \seq \P'')  = \T\sparop \P' \seq \T \sparop \P''
\end{array}$

\smallskip
\noindent
and, in the rest of the paper, we will never distinguish equal lams. For instance,
$\M(\wt{u}) \seq (x,y)$ and $(x,y) \seq \M(\wt{u})$ will be always
identified. These axioms permit to rewrite a lam without function invocations as
a \emph{collection} (operation $\seq$) \emph{of relations} (elements of a relation are gathered by
the operation $\sparop$).

\begin{proposition}
\label{prop.normalform}
For every $\T$, there exist $\T_1, \cdots ,
\T_n$ that are dependencies composed with $\sparop$, such that $\T = \T_1 \seq \cdots 
\seq \T_n$. 
\end{proposition}

\begin{remark}
Lams are intended to be abstract models of programs that highlight the 
resource dependencies in the reachable states.  
The lam $\T_1 \seq \cdots \seq 
\T_n$ of Proposition~\ref{prop.normalform} models a program whose possibly
infinite set of states $\{ \S_1, \S_2, \cdots \}$ is such that the 
resource dependencies in $\S_i$ are a subset of those in some $\T_{j_i}$, with
$1 \leq j_i \leq n$. With this meaning, generic lams 
$\P_1 \seq \cdots \seq \P_m$ are abstractions of transition
systems (a standard model of programming languages), where transitions are
ignored and states record the resource dependencies and the function invocations.
\end{remark}

\begin{remark}
The above axioms, such as
$\T \sparop (\P' \seq \P'')  = \T\sparop \P' \seq \T \sparop \P''$
are restricted to terms $\T$ that do not contain function invocations. In fact,
$\M(\wt{u}) \sparop ((x,y)$ $\seq (y,z)) \neq (\M(\wt{u}) \sparop (x,y)) \seq 
(\M(\wt{u}) \sparop (y,z))$ because the two terms have a different number
of occurrences of invocations of $\M$, and this is crucial for linear recursion -- see
Definition~\ref{def.linearity}.
\end{remark}

In the paper, we always assume lam programs 
$\bigl(\M_1(\wt{x_1}) = \P_{1}, \cdots , 
\M_\ell(\wt{x_\ell}) = \P_{\ell}, \P \bigr)$
to be \emph{well-defined}, namely
\emph{(1)} all function names occurring in $\P_{i}$ and  $\P$ are defined;
\emph{(2)} the arity of function invocations matches that of the corresponding function definition.

\paragraph{Operational semantics.} 
Let a \emph{lam context}, noted $\CP[~]$, be a term derived by the 
following syntax:

\smallskip

$\CP[~] \quad ::= \quad [~] \qquad | \qquad \P \sparop \CP[~] 
\qquad | \qquad \P \seq \CP[~]$

\smallskip

As usual $\CP[\P]$ is the lam where the hole of $\CP[~]$ is replaced by $\P$.
The operational semantics of a program 
$\bigl(\M_1(\wt{x_1}) = \P_{1}, \cdots ,$ 
$\M_\ell(\wt{x_\ell})$ $= \P_{\ell}, \P_{\ell+1} \bigr)$ is a transition system
whose \emph{states}  are pairs $\ilam{\por{V}, \; \P}{}$
and the \emph{transition relation}
is the least one satisfying the rule:

\smallskip

$\begin{array}{c}
\mathrule{Red}{
	\begin{array}{c}
    \M(\wt{x}) = \P
	\qquad 
	\var{\P} \setminus\wt{x}=\wt{z} 
	\qquad \wt{w} \mbox{ are fresh} 
	\\
	\P \subst{\wt{w}}{\wt{z}} \subst{\wt{u}}{\wt{x}} = \P'
	\end{array}
	}{
	\ilam{\por{V}, \; \CP[\M(\wt{u})]}{} 
        \lred{} 
	\ilam{\por{V} \transclosure \wt{u}\lessc \wt{w}, \; \CP[\P']}{}
	}
\end{array}$

\smallskip

By \rulename{red}, a lam $\P$ is evaluated by successively replacing
function invocations with the corresponding lam instances. 
Name creation is handled with a 
mechanism similar to that of mutations. For example, if
$\M(x) = (x,y)\sparop\M(y)$ and $\M(u)$ occurs in the main lam, then
$\M(u)$ is replaced by $(u,v)\sparop\M(v)$, where $v$ is a \emph{fresh maximal
name} in some partial order.
The initial state of a program with main lam $\P$ 
is $\ilam{\por{I}_{\P}, \;
\P}{}$, where $\por{I}_{\P} \eqdef \{(x,x) \; |\; x\in \var{\P}\}$.

To illustrate the semantics of the language of lams we discuss 
three
examples:
\begin{enumerate}
\item
$\bigl( \, %
\M(x,y,z) = (x,y) \sparop \N(y,z) \seq (y,z) , \;
\N(u,v) = (u,v) \seq (v,u)  ,$ 
$\M(x,y,z) \, \bigr)
$
and $\por{I} =
\{ (x,x), (y,y), (z,z)\}$. Then

\smallskip

$\begin{array}{rl}
\ilam{\por{I}, \; \M(x,y,z)}{} \; \lred{} &
\ilam{\por{I}, \; (x,y) \sparop \N(y,z) \seq (y,z)}{}
\\
\lred{} &
\ilam{\por{I}, \; (x,y) \sparop (y,z) \seq \; (x,y) \sparop (z,y) \seq (y,z)}{}
\end{array}$

\smallskip

The lam in the final state \emph{does not contain function invocations}. This 
is because the above program is not recursive.
Additionally, the evaluation of $\M(x,y,z)$ \emph{has not created names}. This 
is because names in the bodies of 
$\M(x,y,z)$ and $\N(u,v)$ are bound.

\item
$\bigl( \M'(x) = (x,y) \sparop \M'(y) \; , \; 
\M'(x) \bigr)$ and $\por{V}_0 = \{ (x_0, x_0) \}$. Then

\smallskip

$\begin{array}{l}
\ilam{\por{V}_0 , \; \M'(x_0)}{} 
\\
\lred{} \, 
\ilam{\por{V}_1 , \; (x_0,x_1) \sparop \M'(x_1)}{}
\\
\lred{} \,
\ilam{\por{V}_2 , \; (x_0,x_1) \sparop (x_1,x_2) \sparop \M'(x_2)}{}  
\\
\lred{}^n 
\ilam{\por{V}_{n+2} , \; (x_0,x_1) \sparop 
\cdots  \sparop (x_{n+1},x_{n+2}) \sparop
\M'(x_{n+2}) }{} 
\end{array}$

\smallskip

\noindent where $\por{V}_{i+1} = \por{V}_i \transclosure x_i \lessc x_{i+1}$.
In this case, the states grow in the number of dependencies as the
evaluation progresses. This growth \emph{is due to the presence of a free name}
in the definition of $\M'$ that, as said, corresponds to generating a fresh
name at every recursive invocation.

\item
$\bigl( \M''(x) = (x,x') \seq (x,x') \sparop \M''(x') , \;
\M''(x_0) \bigr)$ and $\por{V}_0 = \{ (x_0,x_0)\}$. Then

\smallskip

$\begin{array}{l}
\ilam{\por{V}_0 , \; \M''(x_0)}{} 
\\
\quad \lred{} \;
\bigl\langle \por{V}_1 , (x_0,x_1) \seq  (x_0,x_1) \sparop \M''(x_1) \bigr\rangle
\\
\quad \lred{} \;
\bigl\langle \por{V}_2 , (x_0,x_1) \seq  (x_0,x_1) \sparop 
(x_1,x_2) \seq  
(x_0,x_1) \sparop 
(x_1,x_2) \sparop \M''(x_2) \bigr\rangle
\\
\quad \lred{}^n \;
\bigl\langle \por{V}_{n+2}, (x_0,x_1) \seq \cdots \seq
(x_0,x_1) \sparop \cdots \sparop (x_{n+1},x_{n+2}) \sparop
{\M''}(x_{n+2}) \bigr\rangle
\end{array}$

\smallskip

\noindent where $\por{V}_{i+1}$ are as before. 
In this case, the states grow in the number of ``$\seq$''-terms, 
which become
larger and larger as the evaluation progresses.
\end{enumerate}

The semantics of the language of lams is nondeterministic because of the choice of the
invocation to evaluate. However, lams enjoy a diamond
property \emph{up-to bijective renaming of (fresh) names}.

\begin{proposition}
\label{prop.diamond}
Let $\imath$ be a bijective renaming and $\imath(\por{V}) = \{ (\imath(x), \imath(y)) \;
| \; (x,y) \in \por{V} \}$. Let also 
$\ilam{\por{V}, \; \P}{} \lred{} \ilam{\por{V}', \; \P'}{}$ and 
$\ilam{\imath(\por{V}), \; \P \subst{\imath(\wt{x})}{\wt{x}}}{} \lred{} \ilam{\por{V}'', 
\; \P''}{}$, where $\wt{x} = \var{\por{V}}$. Then 
\begin{itemize}
\item[$(i)$] either there exists a bijective renaming $\imath'$ such that 
$\ilam{\por{V}'', \; \P''}{} = \ilam{\imath(\por{V}'), \; \P' \subst{\imath(\wt{x'})}{\wt{x'}}}{}$, where $\wt{x}' = \var{\por{V}'}$,

\item[$(ii)$] or
there exist $\P'''$ 
and a bijective renaming $\imath'$ such that
$\ilam{\por{V}', \; \P'}{} \lred{} \ilam{\por{V}''', \; \P'''}{}$ and 
$\ilam{\por{V}'', \; \P''}{} \lred{} \ilam{\imath'(\por{V}'''), \; \P'''
\subst{\imath'(\wt{z})}{\wt{z}}}{}$, where $\wt{z} = \var{\por{V}'''}$.
\end{itemize}
\end{proposition}

\paragraph{The informative operational semantics.}
In order to detect the circularity-freedom, our technique computes a lam  till every function therein has been adequately unfolded (up-to twice the order of the associated
mutation). This is formalized by switching to an ``informative'' operational 
semantics where basic terms
(dependencies and function invocations) are labelled
by so-called \emph{histories}. 
 
Let a \emph{history}, ranged over by $\alpha, \beta, \cdots$,  be a sequence
 of function names $\M_{i_1} \M_{i_2} \cdots \M_{i_n}$. We write
$\M \in \alpha$ if $\M$ occurs in $\alpha$. 
We
also write $\alpha^n$ for 
$\underbrace{\alpha \cdots \alpha}_{n \; {\rm times}}$.
Let $\alpha \preceq \beta$ if there is $\alpha'$ such that $\alpha \alpha'
= \beta$. The symbol $\varepsilon$ denotes the empty history.

The informative operational semantics is a transition system whose states are
tuples $\lam{\por{V}, \, \suh{\Mplus}, \, \mathbb{L}}{}$ where
$\suh{\Mplus}$ is a set of function invocations with histories and 
$\mathbb{L}$, called \emph{informative lam}, is a  term as $\P$, except that
pairs and function invocations are indexed by
histories, i.e.~$\suhs{\alpha}{(x,y)}$ and $\suhs{\alpha}{\M(\wt{u})}$, respectively. 

Let

\smallskip

$\begin{array}{rc@{\!}l}
\addh{\alpha}{\P} & \eqdef & \left\{
	\begin{array}{l@{\;}l}
	\suhs{\alpha}{(x,y)} & {\rm if} \; {\P}
	 = (x,y) 
	\\[.3cm]
	\suhs{\alpha}{\M(\wt{x})} & 
	{\rm if} \; {\P} = \M(\wt{x}) 
	\\[.3cm]
	\addh{\alpha}{\P'} \sparop \addh{\alpha}{\P''}& 
	{\rm if} \; {\P} = \P' \sparop \P'' 
	\\[.3cm]
	\addh{\alpha}{\P'} \seq \addh{\alpha}{\P''}& 
	{\rm if} \; {\P} = \P' \seq \P'' 
	\end{array} \right.
\end{array}$

\smallskip
\noindent
For example 
$\addh{\M\Q}{(x_4,x_2)\sparop\M(x_2,x_3,x_4,x_5)}=$
$\suhs{\M\Q}{(x_4,x_2)}\sparop$ $\suhs{\M\Q}{\M(x_2,x_3,x_4,x_5)}$.
Let also $\suh{\CP}[~]$ be a lam context with histories (dependency pairs and
function invocations are labelled by histories, the definition is similar to
$\CP[~]$).

The informative transition relation is the least one such that

\smallskip

$\begin{array}{c}
\mathrule{Red+}{
	\begin{array}{c}
	\M(\wt{x})= \P
	\qquad \var{\P}\setminus\wt{x}=\wt{z} 
	\qquad \wt{w} \mbox{ are fresh} %
	\\
	 \P\subst{\wt{w}}{\wt{z}}\subst{\wt{u}}{\wt{x}}=\P'
	\end{array}
	}{
\begin{array}{l}
	\ilam{\mathbb{V}, \, \suh{\Mplus}, \, \suh{\CP}[\suhs{\alpha}{\M(\wt{u}) }]
	 }{} \quad \lred{} 
	 \quad 
	\ilam{\mathbb{V} \transclosure \wt{u} {\lessc} \wt{w}, \; 
	\suh{\Mplus} \cup \{ \suhs{\alpha}{\M(\wt{u})} \}, \, 
	\suh{\CP}[\addh{\alpha\M}{\P'}] }{}
\end{array}
	}
\end{array}$

\smallskip

When $\ilam{\mathbb{V}, \, \suh{\Mplus}, \, \mathbb{L}}{} \lred{} \ilam{\mathbb{V}', \, \suh{\Mplus}', \,
\mathbb{L}'}{}$ by applying \rulename{Red+} to $\suhs{\alpha}{\M(\wt{u})}$,
we say that the term $\suhs{\alpha}{\M(\wt{u})}$ \emph{is evaluated in the 
reduction}. The initial informative state of a program with main lam
 $\P$ is $
\ilam{\por{I}_{\P}, \, \varnothing, \, \addh{\varepsilon}{\P} }{}$.

For example, the $\M\Q\PP$-program

\smallskip

$\begin{array}{llrl}
\bigl( \; &\M(x,y,z,u) &=&  (x,z)\sparop\Q(u,y,z)  \; ,\\
 &\Q(x,y,z) &=&  (x,y)\sparop\M(y,z,x,u)\;,\\
 &\PP(x,y,z,u)&=& (z,x)\sparop\PP(x,y,z,u) \sparop \M(x,y,z,u)\;,
 \\
 & \PP(x_1,x_2,x_3,x_4) \! \! &  
 \bigr) &
\end{array}$

\smallskip

\noindent has an (informative) evaluation

\smallskip

$\begin{array}{l}
\ilam{\por{I}_{\P},\, \emptyset, \, 
\suhs{\varepsilon}{\PP(x_1,x_2,x_3,x_4)} }{}
\\
\; \lred{} \;
\ilam{\por{I}_{\P},\, \suh{\Mplus}, \, \mathbb{L}\,\sparop\,\suhs{\PP}{\M(x_1,x_2,x_3,x_4)} }{}
\\
\; \lred{} \;
\ilam{\por{I}_{\P},\, \suh{\Mplus}_1, \, \mathbb{L}\,\sparop\,\suhs{\PP\M}{(x_1,x_3)}\,\sparop\,\suhs{\PP\M}{\Q(x_4,x_2,x_3)} }{}
\\
\; \lred{} \;
\ilam{\por{I}_{\P}{{\transclosure x_4\lessc x_5}},\, \suh{\Mplus}_2, \,
\mathbb{L}'\,\sparop\,\suhs{\PP\M\Q}{(x_4,x_2)}\,\sparop\,\suhs{\PP\M\Q}{\M(x_2,x_3,x_4,x_5)} }{}
\\
\end{array}$

\smallskip 

\noindent where $\mathbb{L}=\suhs{\PP}{(x_3,x_1)}\sparop\suhs{\PP}{\PP(x_1,x_2,x_3,x_4)}$, $\mathbb{L}'=\mathbb{L}\sparop\suhs{\PP\M}{(x_1,x_3)}$ and $\suh{\Mplus}=\{\suhs{\varepsilon}{\PP(x_1,x_2,x_3,x_4)}\}$, 
 $\suh{\Mplus}_1=\suh{\Mplus}\cup\{\suhs{\PP}{\M(x_1,x_2,x_3,x_4)}\}$,
 $\suh{\Mplus}_2= \suh{\Mplus}_1\cup\{\suhs{\PP\M}{\Q(x_4,x_2,x_3)}\}$.

There is a strict correspondence between the non-informative and informative semantics  
that is crucial for the correctness of our algorithm in 
Section~\ref{sec.thealgorithm-vero}.
Let $\sem{\cdot}$ be an \emph{eraser map} that takes an informative lam and 
removes the histories. 
The formal definition is omitted because
it is straightforward.

\begin{proposition}
\begin{enumerate}
\item
If $\ilam{\mathbb{V}, \suh{\Mplus}, \, \mathbb{L}}{} \lred{} \ilam{\mathbb{V}', \suh{\Mplus}', \, \mathbb{L}'}{}$
then $\ilam{\mathbb{V}, \sem{\mathbb{L}}}{}$ $\lred{} \ilam{\mathbb{V}', \sem{\mathbb{L}'}}{}$;
\item
If $\ilam{\mathbb{V}, \sem{\mathbb{L}}}{} \lred{} \ilam{\mathbb{V}', \P'}{}$ then 
there are $\suh{\Mplus}$, $\suh{\Mplus}'$, $\mathbb{L}'$ such that $\sem{\mathbb{L}'} = \P'$ and
$\ilam{\mathbb{V}, \suh{\Mplus}, \mathbb{L}}{} \lred{} \ilam{\mathbb{V}', \suh{\Mplus}', \mathbb{L}'}{}$.
\end{enumerate}
\end{proposition}

\paragraph{Circularities.}
Lams record sets of relations on names. The following function
$\flatt{\cdot}$, called \emph{flattening}, makes explicit 
these relations 
\[
\begin{array}{c}
\flatt{\pinull} =  \pinull,
\qquad
\flatt{(x,y)}  =  (x,y),
\qquad
\flatt{\M(\wt{x})}  =  \pinull,
\\
\flatt{\P \sparop \P'}  =  \flatt{\P} \sparop \flatt{\P'} ,
\qquad
\flatt{\P \seq \P'}  =  \flatt{\P} \seq \flatt{\P'}.
\end{array}
\]
For example, if
$\P = \M(x,y,z) \seq (x,y) \sparop \N(y,z) \sparop \M(u,y,z) 
 \seq \N(u,v) \sparop (u,v) \seq (v,u)$ then
$\flatt{\P} =  (x,y) \seq (u,v) \seq (v,u)$. 
That is, there are three relations in $\P$:
$\{ (x,y)\}$ and $\{ (u,v)\}$ and $\{ (v,u) \}$. 
By Proposition~\ref{prop.normalform},
$\flatt{\P}$ returns, up-to the lam axioms, sequences of (pairwise different)
$\sparop$-compositions of dependencies. 
The operation $\flatt{\cdot}$ may be extended to informative lams $\mathbb{L}$ in
the obvious way:
$\flatt{\suhs{\alpha}{(x,y)}}  =  \suhs{\alpha}{(x,y)}$ and $
\flatt{\suhs{\alpha}{\M(\wt{x})}}  =  \pinull$.

\begin{definition}
\label{def.circularity}
A lam $\P$ \emph{has a circularity} if 

\smallskip

$\qquad\flatt{\P} =  
(x_1,x_2) \sparop (x_2,x_3) \sparop  \cdots \sparop  (x_m,x_1) \sparop \T' 
\seq \T''$

\smallskip

\noindent for some $x_1, \cdots, x_m$.
A state $\ilam{\por{V}, \; \P}{}$ has a circularity if $\P$ has a circularity.
Similarly for an informative lam $\mathbb{L}$.
\end{definition}
The final state of  the $\M\N\PP$-program computation has a circularity;
another function displaying a circularity is $\N$ in  
Section~\ref{sec.introduction}. 
None of the states in the examples 1, 2, 3 at the beginning of this section
has a \emph{circularity}.

\section{Linear recursive lams and saturated states}
\label{sec.algorithm}

This section develops the theory that underpins the algorithm of 
Section~\ref{sec.thealgorithm-vero}. In order to lightening the section,
 the technical details
have been moved in Appendix~\ref{sec.thetheory-nfp}.

We restrict our arguments to (mutually) recursive lam programs. In fact,
circularity analysis in non-recursive programs is trivial: it is sufficient to evaluate 
all the invocations till the final state and verify the presence of circularities therein.
A further restriction allows us to simplify the arguments without loosing in
generality (\emph{cf.}~the definition of saturation): we assume that every function is (mutually) recursive. We may reduce
to this case by expanding function invocation of non-(mutually) recursive functions
(and removing their definitions). 

\paragraph{Linear recursive functions and mutations.}
Our decision algorithm relies on interpreting recursive functions as mutations. 
This interpretation is not always possible: the recursive functions that have an associated
mutation are the linear recursive ones, as defined below. 

The technique for dealing with the general case is briefly discussed in 
Section~\ref{sec.relatedworks} and is detailed in Appendix~\ref{sec.nonlinear}.

\begin{definition}
\label{def.linearity}
Let $\bigl(\M_1(\wt{x_1}) = \P_{1}, \cdots , 
\M_\ell(\wt{x_\ell}) = \P_{\ell}, \P \bigr)$ be a lam program.
A sequence $\M_{i_0} \M_{i_1} \cdots
\M_{i_k}$ is called a \emph{recursive history of} $\M_{i_0}$ if
(a) the function names are pairwise different and
(b) for every $0 \leq j \leq k$, $\P_{i_j}$ contains one invocation of
$\M_{i_{j+1 \% k}}$ (the operation $\%$ is the remainder of the division).

The lam program is \emph{linear recursive} if (a) every
function name has a unique recursive history and (b) if 
$\M_{i_0} \M_{i_1} \cdots \M_{i_k}$ is a recursive history then,
for every $0 \leq j \leq k$, $\P_{i_j}$ contains \emph{exactly} one invocation of
$\M_{i_{j+1 \% k}}$. 
\end{definition}
For example, the program

\smallskip
 
$\begin{array}{ll}
\bigl( & \M_1(x,y) = (x,y) \sparop \M_1(y,z) \sparop \M_2(y) \seq
\M_2(z) \; , 
\M_2(y) = (y,z)\sparop\M_2(z) \; ,  
\P \quad \bigr)
\end{array}$

\smallskip

\noindent is linear recursive. On the contrary

\smallskip 

$\bigl( \M(x) = (x,y) \sparop \N(x) \; , \; \N(x) = (x,y) \sparop \M(x) \seq \N(y) \; ,   \;\P \; \bigr)$

\smallskip

\noindent is not linear recursive because $\N$ has two recursive histories, namely $\N$ and $\N\M$.

Linearity allows us to associate a \emph{unique} mutation to every function name.
To compute this mutation, let 
$\mathtt{H}$ range over sequences of function invocations.
We use the following two rules:
\[
\begin{array}{c}
\bigfract{
	\M_i \alpha \models \varepsilon \quad  \M_{i}(\wt{x}_i) = \P_{i}
	}{
	\alpha \models \M_i(\wt{x}_i)
}
\quad
\bigfract{
\begin{array}{c}
	\M_{j} \alpha  \models {\tt H}
	\M_{i}(\wt{x})
	\quad  
	\M_{i}(\wt{x}_i) = \P_{i}
	\\ 
	\var{\P_{i}}\setminus\wt{x}_i =\wt{z} 
	\qquad \wt{w} \mbox{ are fresh} 
	\\
	\CP[\M_{j}(\wt{y})] = \P_{i} \subst{\wt{w}}{\wt{z}}\subst{\wt{x}}{\wt{x}_i}
	\end{array}
	}{
	\alpha \models 
	{\tt H} \M_{i}(\wt{x}) \M_{j}(\wt{y})
}
\end{array}
\]

Let $\varepsilon \models \M(x_1, \cdots, x_n) \cdots
\M(x_1', \cdots , x_n')$ be the final judgment of the proof tree
with leaf $\M \alpha \M \models \varepsilon$, where $\M\alpha$ is the 
recursive history of $\M$.
Let also
$x_1', \cdots, x_n' \setminus x_1, \cdots, x_n = z_1, \cdots , z_k$. Then
the \emph{mutation of} $\M$, written $\mu_{\M} = \mut{a_1, \cdots, a_n}$ is defined by 

\smallskip

$\qquad\qquad a_i \; = \; \left\{ 
	\begin{array}{l@{\qquad}l}
	j & {\rm if} \; x_i' = x_j
	\\
	\\
	n+j & {\rm if} \; x_i' = z_j
	\end{array} \right.$

\smallskip

Let $\ordermut{\M}$, called \emph{order of the function} $\M$, 
be the order of $\mu_\M$. 
For example, in the $\M\Q\PP$-program,
the recursive history of $\M$ is $\M \Q$ and, applying the algorithm above
to $\M \Q \M  \models \varepsilon$, we get 
$\varepsilon \models  \M(x,y,z,u) \Q(u,y,z)\M(y,z,u,v)$. 
The mutation of $\M$ is $\mut{2,3,4,5}$ and $\ordermut{\M}=4$.
Analogously we can compute $\ordermut{\Q}=3$ and $\ordermut{\PP}=1$.

\medskip
\emph{Saturation.}
In the remaining part of the section we assume a fixed linear recursive program 
$\bigl(\M_1(\wt{x_1}) = \P_{1}, \cdots , 
\M_\ell(\wt{x_\ell}) = \P_{\ell}, \P \bigr)$ and let 
$\ordermut{\M_1}, \cdots, \ordermut{\M_\ell}$ be the orders of the corresponding
functions. 

\begin{definition}
A history $\alpha$ is 
\begin{description}
\item[$\M$-\emph{complete}] 
\ \\
if $\alpha = \beta^{\ordermut{\M}}$, where $\beta$ is the recursive
history of $\M$. %
We say that $\alpha$ is \emph{complete} when it is 
$\M$-complete, for some $\M$.

\item[$\M$-\emph{saturating}] 
\ \\
if $\alpha = \beta_1 \cdots 
\beta_{n-1} \alpha_n^{2}$, where $\beta_i \preceq (\alpha_i)^2$, with $\alpha_i$ 
complete, and $\alpha_n$ $\M$-complete.
We say that $\alpha$ is \emph{saturating} when it is 
$\M$-saturating, for some $\M$.
\end{description}
\end{definition}
In the $\M\Q\PP$-program,  $\ordermut{\M}=4$,
$\ordermut{\Q}=3$, and $\ordermut{\PP}=1$, and
the recursive histories of $\M$, $\Q$ and $\PP$ are equal to $\M\Q$, to 
$\Q\M$ and to $\PP$, respectively.
Then $\alpha= (\M\Q)^{4}$ is the $\M$-complete history
 and $\PP^2(\M\Q)^8$ and $\PP(\M\Q)^8$ are $\M$-saturating.

The following proposition is an important consequence of the theory of mutations
(Theorem~\ref{thm.mainthm}) and the semantics of lams (and their axioms).
In particular, it states that, 
if a function invocation $\M_0(\wt{u_0})$ is unfolded up to the order of $\M_0$ then (i) the last invocation $\M_0(\wt{v})$ may be 
mapped back to a previous invocation by a flashback and (ii) the same
flashback also maps back dependencies created by the unfolding of $\M_0(\wt{v})$.

\begin{proposition}
\label{prop.flashback-order}
Let  $\beta=\M_0 \M_1 \cdots \M_n$ be $\M_0$-complete
and let 
\[
\begin{array}{l}
\ilam{\mathbb{V}, \, \suh{\Mplus}, \, 
\suh{\CP}_0[\suhs{\alpha}{\M_0(\wt{u_0})}]}{} \lred{}^{n+1} \;
\ilam{\mathbb{V}', \, \suh{\Mplus}', \, 
\suh{\CP}_0[\suh{\CP}_1[ \cdots \suh{\CP}_n[\suhs{\alpha\M_0 \cdots \M_n}{\M_0(\wt{u_{n+1}})}] \cdots]]}{}
\end{array}
\]
where $\suh{\Mplus}' = \suh{\Mplus} \cup \{ \suhs{\alpha}{\M_0(\wt{u_0})},
\suhs{\alpha\M_0}{\M_1(\wt{u_1})}, \cdots, \suhs{\alpha\M_0 \cdots \M_{n-1}}{\M_n(\wt{u_n})}\}$ and $\M_i(\wt{u_i}) = \P_i'$ and $\addh{\alpha\M_0 \cdots \M_{i}}{\P_i'} = 
\suh{\CP}_i[\suhs{\alpha\M_0 \cdots \M_{i}}{\M_{i+1}(\wt{u_{i+1}})}]$
(unfolding of the functions in the complete history of $\M_0$). Then
there is a $\suhs{\alpha\M_0 \cdots \M_{h-1}}{\M_h(\wt{u_h})} \in \suh{\Mplus}'$ 
and a $\por{V}'$-flashback $\rho$ such that
\begin{enumerate}
\item
$\M_0(\rho(\wt{u_{n+1}})) = \M_h(\wt{u_h})$ (hence $\M_0 = \M_h$);
\item
let $\M_0(\wt{u_{n+1}}) = \P$ and $\flatt{\P} = \T_1 \seq \cdots \seq \T_k$
and 
\\
$\flatt{\suh{\CP}_0[\suh{\CP}_1[ \cdots \suh{\CP}_n[\suhs{\alpha\beta}{\M_0(\wt{u_{n+1}})}] \cdots]]} = \suh{\T_1'} \seq \cdots \seq \suh{\T_{k'}'}$. Then, for every
$1 \leq i \leq k$, there exists $1 \leq j \leq k'$ such that 
$\suh{\T_j'} = \addh{\alpha\M_0 \cdots \M_{h-1}}{\T_i} \sparop \suh{\T_j''}$, for
some $\suh{\T_j''}$.
\end{enumerate}
\end{proposition}

The notion of $\M$-saturating will be used to define a
``saturated'' state, i.e., a state where the evaluation of programs
may safely (as regards circularities) stop.

\begin{definition}
\label{def.saturation}
An informative lam $\ilam{\mathbb{V}, \, \suh{\Mplus} , \,
\mathbb{L}}{}$ 
is \emph{saturated} when, for every $\suh{\CP}[~]$ and $\M(\wt{u})$ such that
$\mathbb{L} = \suh{\CP}[\suhs{\alpha}{\M(\wt{u})}]$, 
$\alpha$ has a saturating prefix.
\end{definition}

It is easy to check that the following informative lam generated by the computation of the {$\M\Q\PP$}-program is saturated:

\smallskip

$\begin{array}{ll}
\bigl\langle
\mathbb{V}_7,\, 
\suh{\Mplus}, \, \suhs{\PP^2}{\PP(x_1,x_2,x_3,x_4)}
 &\sparop\, \prod_{0\leq i\leq 8}\suhs{\PP\M(\Q\M)^i}{(x_{i+1},x_{i+3})}
 \\
\,&\sparop\,\prod_{0\leq i\leq 8} \suhs{\PP(\M\Q)^i}{(x_{i+3},x_{i+1})}
\\
\,&\sparop\,\suhs{\PP(\M\Q)^8}{\M(x_9,x_{10},x_{11},x_{12})} 
\bigr\rangle,
\end{array}$

\smallskip
 
where $\mathbb{V}_{i+1}=\mathbb{V}_i\transclosure x_{i+4} {\lessc} x_{i+5}$, and 

\smallskip

$\begin{array}{ll}
\suh{\Mplus} = & 
 \{\suhs{\varepsilon}{\PP(x_1,x_2,x_3,x_4)}, \suhs{\PP}{\PP(x_1,x_2,x_3,x_4)} \}
 \\
 & \cup\{ \suhs{\PP(\M\Q)^i}{\M(x_{i+1},x_{i+2},x_{i+3},x_{i+4})} \; | \;
 {0\leq i\leq 7} \}
   \\ 
 & {\cup} \; \{\suhs{\PP\M(\Q\M)^i}{\Q(x_{i+4},x_{i+2},x_{i+3})} \; | \;
 {0\leq i\leq 7} \}.
\end{array}$

\smallskip

Every preliminary statement is in place for our key theorem that details the mapping
of circularities created by transitions of saturated states to past circularities.

\begin{theorem}
\label{thm.invariance}
Let $\ilam{\por{I}_{\P}, \, \varnothing, \,  \addh{\varepsilon}{\P}}{} \lred{}^* \ilam{\mathbb{V}, \, \suh{\Mplus}, \, \mathbb{L}}{}$ 
and
$\ilam{\mathbb{V}, \,  \suh{\Mplus}, \,\mathbb{L}}{}$ be a saturated state.
If
$\ilam{\mathbb{V},  \suh{\Mplus}, \,\mathbb{L}}{} \lred{}
\ilam{\mathbb{V}', \,  \suh{\Mplus}', \,\mathbb{L}'}{}$ then
\begin{enumerate}
\item
$\ilam{\mathbb{V}', \,  \suh{\Mplus}', \,\mathbb{L}'}{}$ is saturated;
\item
if $\mathbb{L}'$ has a circularity then $\mathbb{L}$ has already 
a circularity.
\end{enumerate}
\end{theorem}

\begin{proof}
(\emph{Sketch})
Item 1. directly follows from Proposition~\ref{prop.flashback-order}.
However, this proposition is 
not sufficient to guarantee that circularities created in saturated states are mapped
back to past ones. In particular, the interesting case is the one of 
\emph{crossover circularities}, as discussed in Section~\ref{sec.introduction}.
Therefore, let 

\smallskip

$\suhs{\alpha_1}{(x_1, x_{2})}, \cdots , 
\suhs{\alpha_{h-1}}{(x_{h-1},x_{h})}, 
\suhs{\alpha_h}{(x_h,x_{h+1})}, \cdots , \suhs{\alpha_{n}}{(x_{n},x_{1})}$

\smallskip

\noindent be a circularity in $\mathbb{L}'$ such that 
$\suhs{\alpha_h}{(x_h,x_{h+1})}, \cdots , \suhs{\alpha_{n}}{(x_{n},x_{1})}$
were already present in $\mathbb{L}$. 
Proposition~\ref{prop.flashback-order} guarantees the existence of a flashback 
$\rho$ that maps 
 $\suhs{\alpha_1}{(x_1, x_{2})}$ $\sparop \cdots  \sparop \suhs{\alpha_{h-1}}{(x_{h-1},x_{h})}$ to
$\suhs{\alpha_1}{(\rho(x_1), \rho(x_{2}))} \sparop \cdots \sparop 
\suhs{\alpha_{h-1}}{(\rho(x_{h-1}), \rho(x_{h}))}$. However, it is possible that

\smallskip

$\begin{array}{r}
\suhs{\alpha_1}{(\rho(x_1), \rho(x_{2}))} \sparop \cdots \sparop
\suhs{\alpha_{h-1}}{(\rho(x_{h-1}), \rho(x_{h}))} %
\sparop 
\suhs{\alpha_h}{(x_h,x_{h+1})}\sparop \cdots \sparop \suhs{\alpha_{n}}{(x_{n},x_{1})}
\end{array}$

\smallskip

\noindent is no more a circularity because, for example, $\rho(x_{h}) \neq x_{h}$ (assume
that $\rho(x_{1}) = x_{1}$). 
Let us discuss this issue. The hypothesis of saturation
guarantees that transitions produce histories $\alpha^2\beta$, where $\alpha$ is
 complete. Additionally, $\alpha_1, \cdots, \alpha_{h-1}$ 
must be equal because they have been created by
$\ilam{\mathbb{V},  \suh{\Mplus}, \,\mathbb{L}}{}$ $\lred{}
\ilam{\mathbb{V}', \,  \suh{\Mplus}', \,\mathbb{L}'}{}$. For simplicity, let 
 $\beta = \M$ and $\alpha = \M\alpha'$. Therefore, by 
Proposition~\ref{prop.flashback-order}, $\rho$ maps 
 $\suhs{\alpha^2\M}{(x_1, x_{2})} \sparop \cdots \sparop$ $ \suhs{\alpha^2\M}{(x_{h-1},x_{h})}$ to
$\suhs{\alpha\M}{(\rho(x_1), \rho(x_{2}))} \sparop \cdots \sparop \suhs{\alpha\M}{(\rho((x_{h-1}),
\rho(x_{h}))}$ and, $\rho(x_h) \neq x_h$
when $x_{h}$ is created by the computation evaluating functions in $\alpha'$. 

To overcome this problem, it is possible to demonstrate using a statement similar to
(but stronger than) Proposition~\ref{prop.flashback-order} that
$\rho$ maps 
$\suhs{\alpha_h}{(x_h,x_{h+1})}$ $\sparop \cdots \sparop \suhs{\alpha_{n}}{(x_{n},x_{1})}$
to
$
\suhs{[\alpha_h]}{(\rho(x_h),\rho(x_{h+1}))}\sparop \cdots \sparop
 \suhs{[\alpha_{n}]}{(\rho(x_{n}),\rho(x_{1}))}
$
where $[\alpha_i]$ are ``kernels'' of $\alpha_i$ where every $\gamma^k$ in 
$\alpha_i$, with 
$\gamma$ a complete history and $k \geq 2$, is replaced by $\gamma$.
The proof terminates  by demonstrating that the term

\smallskip

$\begin{array}{l}
\suhs{\alpha\M}{(\rho(x_1), \rho(x_{2}))} \sparop \cdots \sparop \suhs{\alpha\M}{(\rho((x_{h-1}),
\rho(x_{h}))}
\\
\qquad 
\sparop \suhs{[\alpha_h]}{(\rho(x_h),\rho(x_{h+1}))}\sparop \cdots \sparop
 \suhs{[\alpha_{n}]}{(\rho(x_{n}),\rho(x_{1}))}
\end{array}$

\smallskip

\noindent is in $\mathbb{L}$ (and it is a circularity).
\end{proof}

\section{The decision algorithm for detecting circularities in linear recursive lams}
\label{sec.thealgorithm-vero}

The algorithm for 
deciding the circularity-freedom problem in linear recursive lam programs
takes as input a lam program
$\bigl(\M_1(\wt{x_1}) = \P_{1}, \cdots , 
\M_\ell(\wt{x_\ell}) = \P_{\ell}, \P \bigr)$ and performs the following steps:

\medskip

\noindent
{\sc Step 1}: \emph{find recursive histories}. By parsing the lam program we create
a graph where nodes are function names and, for every invocation of $\N$ in the 
body of $\M$,  there is an edge from $\M$ to $\N$. Then a standard depth first search
associates to every node its recursive histories (the paths starting and ending at 
that node, if any). The lam program is linear recursive if every node has at most one associated
recursive history.

\medskip

\noindent
{\sc Step 2}: \emph{computation of the orders}. Given the recursive history $\alpha$ 
associated to a function $\M$, we compute the corresponding mutation by running 
$\alpha \models \varepsilon$ (see Section~\ref{sec.algorithm}). A straightforward 
parse of the mutation returns the set of cycles and sinks and, therefore, gives
the order $\ordermut{\M}$.

\medskip

\noindent
{\sc Step 3}: \emph{evaluation process}. The main lam is unfolded till the 
the saturated state. That is, every function invocation $\M(\wt{x})$ in the main lam is
evaluated up-to twice the order of the corresponding mutation. The function
invocation of $\M$ in the saturated state is erased and the process is repeated on 
every other function invocation (which, therefore, does not belong to the recursive 
history of $\M$), till no function invocation is present in the state.
At this stage we use the lam axioms that yield a term $\T_1 \seq \cdots \seq \T_n$. 

\medskip

\noindent
{\sc Step 4}: \emph{detection of circularities}. 
Every $\T_i$ in $\T_1 \seq \cdots \seq \T_n$ may be represented as a graph where
nodes are names and edges correspond to dependency pairs. To detect whether $\T_i$ contains a circular 
dependency, we run Tarjan algorithm~\cite{Tarjan72} for connected components of graphs and we stop
the algorithm when a circularity is found. 

\medskip

Every preliminary notion is in place for stating our main result; we also
make few remarks about the 
correctness of the algorithm and its computational cost. 

\begin{theorem}
The problem of the circularity-freedom of a lam program is decidable when the
program is linear recursive.
\end{theorem}

The algorithm consists of the four steps described above. The critical step, as
far as correctness is concerned, is the third one, which
follows by 
Theorem~\ref{thm.invariance} and by the diamond property in 
Proposition~\ref{prop.diamond} (whatever other computation may be completed in 
such a way the final state is equal up-to a bijection to a saturated state).

As regards the %
computational complexity %
{\sc Steps} 1 and 2 are linear with respect to the size of the
lam program and {\sc Step} 4 is linear with respect to the size of the term
$\T_1 \seq \cdots \seq \T_n$. {\sc Step} 3 evaluates the
program till the saturated state. Let
\begin{description}
\item[$\ordermut{{\it max}}$] be
the largest order of a function;
\item[$m_{\it max}$] be the maximal number of function invocations
in a body, apart the one in the recursive history.
\end{description}
Without loss of generality, we assume that recursive histories have length 1 and that
the main lam consists of $m_{\it max}$ invocations of the same function. 
Then an upper bound to the length of the evaluation till the saturated state is
\[
(2 \times \ordermut{{\it max}} \times m_{\it max}) + (2 \times
\ordermut{{\it max}} \times m_{\it max})^2 + \cdots +  (2 \times
\ordermut{{\it max}} \times m_{\it max})^\ell
\]
Let $k_{\it max}$ be the maximal number of dependency pairs in a body. Then
the size of the saturated state is $O(k_{\it max} \times 
{(\ordermut{{\it max}} \times m_{\it max})^\ell})$, which is also the computational
complexity of our algorithm.

\section{Assessments}
\label{sec.assessments}

The algorithm defined in Section~\ref{sec.thealgorithm-vero} has been 
prototyped
~\cite{DAT}.
The prototype is called {\tt DAT} (Deadlock Analysis Tool).
As anticipated in Section~\ref{sec.introduction}, our analysis 
has been applied to a concurrent object-oriented language called {\tt ABS}~\cite{ABS}, which is a {\sc Java}-like language with futures and an asynchronous concurrency model ({\tt ASP}~\cite{ASP} is another language in the same family). 
The derivation of lams from {\tt ABS} programs is defined by
an \emph{inference system} that has been 
previously developed for {\tt SDA}~\cite{GL2013a}, an integrated deadlock analyzer 
of the {\tt ABS} tool suite. %
The inference system extracts behavioral types from {\tt ABS} programs and feeds them to the analyzer. These types display 
the resource dependencies and the method invocations while discarding irrelevant (for the
deadlock analysis) details. There are two relevant differences between inferred types
and lams: (i) methods' arguments have a record structure and (ii) behavioral types
have the union operator (for modeling the if-then-else statement).
To bridge this gap and have some initial assessments, we perform a basic
\emph{automatic} transformation of types into lams.  

We tested {\tt DAT} on a number of 
medium-size programs written for
benchmarking purposes by {\tt ABS} programmers and on 
an industrial case study based on the Fredhopper Access Server (FAS) developed by SDL Fredhoppper~\cite{FAS}. This Access Server 
provides search and merchandising IT services to e-Commerce companies.
The (leftmost three columns of the) following table reports the experiments: for every program we display the
number of lines, whether the analysis has reported a deadlock ({\tt D}) or not
($\checkmark$), the time in seconds required for the analysis.
Concerning time, we only report the time of the analysis (and not 
the one taken by the inference) 
when they run on a QuadCore 2.4GHz and Gentoo (Kernel 3.4.9):
{\small
\begin{center}
\begin{tabular}{||l||c|@{\quad}c@{\quad}|@{\quad}c@{\quad}|@{\quad}c@{\quad}||}
\hline
\quad program \quad & lines &  \begin{tabular}{c} {\tt DAT} \\ result \ time 
\end{tabular} 
& \begin{tabular}{c} {\tt SDA} \\ result \ time \end{tabular} & \begin{tabular}{c} {\tt DECO} \\
result \ time \end{tabular}
\\ \hline
{\tt PingPong} & 61 &  \checkmark \quad 0.311 & \checkmark \quad 0.046 & \checkmark \quad 1.30 
\\ \hline
{\tt MultiPingPong} & 88  & {\tt D} \quad 0.209 & {\tt D} \quad 0.109 & {\tt D} \quad 1.43 
\\ \hline
{\tt BoundedBuffer}  \quad &  103  & \checkmark \quad 0.126  & \checkmark \quad 0.353 & \checkmark \quad 1.26
\\ \hline
{\tt PeerToPeer}    &   185 & \checkmark \quad 0.320 & \checkmark \quad 6.070 & \checkmark \quad 1.63
\\ \hline \hline
{\tt FAS Module} & 2645 & \checkmark \quad 31.88 & \checkmark \quad 39.78 & \checkmark \quad 4.38
\\
\hline
\end{tabular}
\end{center}
}

The rightmost two columns of the above table reports the results of two other
tools that have also been developed for the deadlock analysis of {\tt ABS} 
programs: 
{\tt SDA}~\cite{GL2013a} and {\tt DECO}~\cite{Antonio2013}. 
The technique used in~\cite{GL2013a} (that underpins the {\tt SDA} tool) derives 
the dependency graph(s) of lam programs
by means of a standard fixpoint analysis. 
To circumvent the issue of the 
infinite generation of new names, the fixpoint is computed on models with a 
limited capacity of name creation. This introduces overapproximations that in turn
display 
false positives (for example, {\tt SDA} returns a false positive for the 
lam of {\tt factorial}). In the present work, this
limitation of finite models is overcome (for linear recursive programs) 
by recognizing patterns of 
recursive behaviors, so that it is possible to 
reduce the analysis to a finite portion of computation without losing precision in 
the detection of deadlocks.
The technique in~\cite{Antonio2013} integrates a point-to analysis with 
an analysis returning (an over-approximation of) program points that may be running 
in parallel. As for other model checking techniques, the authors use 
a finite amount of (abstract) object names to ensure termination
of programs with object creations
underneath iteration or recursion. For example, {\tt DECO} 
 (as well as {\tt SDA}) signals a deadlock in programs containing 
methods whose lam is
\footnote{The code of a corresponding {\tt ABS} program is available at the {\tt DAT} tool website~\cite{DAT}, \emph{c.f.} {\tt UglyChain.abs}.}
$
{\tt m}(x,y) =  (y,x)\sparop {\tt m}(z,x) \; 
$
that our technique correctly recognizes as deadlock-free.

As highlighted by the above table, the three tools return the 
same results as regards deadlock analysis, but are different as
regards performance. In particular {\tt DAT} and {\tt SDA} are comparable on 
small/mid-size programs,  {\tt DECO} appears less performant (except 
for {\tt PeerToPeer}, where 
{\tt SDA} is quite slow because of
the number of dependencies produced by the fixpoint algorithm). 
On the {\tt FAS module},
{\tt DAT} and {\tt SDA} are again comparable -- their computational complexity is
exponential -- {\tt DECO} is more performant because its worst case complexity is
cubic in the dimension of the input. As we discuss above, this gain in performance is payed by {\tt DECO} in a loss of precision.

Our final remark is about the proportion between linear recursive functions 
and nonlinear ones in programs. This is hard to assess and our answer is
perhaps not enough adequate. We have parsed the three case-studies developed in the
European project HATS~\cite{FAS}. The case studies are the FAS module, 
a Trading System (TS) modeling a supermarket handling sales, and a Virtual Office 
of the Future (VOF) where office workers are enabled to perform their office tasks seamlessly independent of their current location. 
FAS has 2645 code-lines, TS has 1238 code-lines, and VOF has 429 code-lines. In none of
them we found a nonlinear recursion, TS and VOF have respectively 2 and 3 
linear recursions (there are
recursions in functions on data-type values that have nothing to do with locks and
control). This substantiates the usefulness of our technique in these programs;
the analysis of a wider range of programs is matter of future work.  

\section{Related works}
\label{sec.relatedworksserious}

The solutions in the literature for deadlock detection in infinite state programs 
either give imprecise answers or do not scale when, for instance, programs also admit 
dynamic resource creation.
Two basic techniques are used: type-checking and 
model-checking.

Type-based deadlock analysis has been extensively studied both for process 
calculi~\cite{Kobayashi06,Suenaga08,Vasconcelos2009} and for
object-oriented programs~\cite{Boyapati2002,Flanagan03,Abadi2006}. 
In Section~\ref{sec.introduction} 
we have thoroughly discussed our position with respect to 
Kobayashi's works; therefore we omit here any additional comment.
In the other contributions about deadlock analysis, a type system computes a 
partial order of 
the deadlocks in a program and a subject reduction theorem proves that 
tasks follow this order. On the contrary, our technique does not
compute any ordering of deadlocks, thus being more flexible: a
computation may acquire two deadlocks in different order at different 
stages, thus being correct in our case, but incorrect with the other techniques.
A further difference with the above works is that we use behavioral types, which
are terms in some simple process algebras~\cite{LP}. The use of simple
process algebras to guarantee the correctness (= deadlock freedom) of
interacting parties is not new. This is the case of the exchange patterns in 
 {\sc ssdl}~\cite{SSDLMEP},
 which are based on CSP~\cite{CSP} and pi-calculus~\cite{PIC}, of session types~\cite{GayN06},
or of 
the terms 
in~\cite{NIELSON} and~\cite{TYPESMODELS}, which use CCS~\cite{CCS}. In these proposals,
the deadlock freedom follows by checking either a dual-type relation or a behavioral
equivalence, which amounts to model checking deadlock freedom on the types. 

As regards model checking techniques, in~\cite{CM97} circular dependencies among processes are detected as erroneous configurations, but dynamic creation of names is not treated. 
An alternative model checking technique is proposed in~\cite{BouajjaniE12} for
multi-threaded asynchronous communication languages with futures (as {\tt ABS}). 
This technique is based on vector systems and addresses infinite-state 
programs that admit 
thread  creation but not dynamic resource creation. 

The problem of verifying deadlocks in infinite state models has been 
studied in other contributions. For example, \cite{Schroter00} compare a
 number of unfolding algorithms for Petri Nets
with techniques for safely cutting potentially infinite unfoldings. 
Also in this work, dynamic resource creation is not addressed.
The techniques conceived for dealing with dynamic name creations are the so-called
\emph{nominal techniques}, such as nominal automata~\cite{Segoufin06,NevenSV01}
that recognize languages over infinite alphabets  and  
HD-automata~\cite{MontanariPistore}, where names are explicit part of the
 operational model.
In contrast to our approach, the models underlying these techniques are finite state. 
Additionally, the dependency relation between names, which is crucial for 
deadlock detection, is not studied.

\section{Conclusions and future work}
\label{sec.relatedworks}

We have defined an algorithm for the detection of deadlocks in infinite 
state programs, which is a decision procedure for linear recursive programs that 
feature dynamic resource creation. 
This algorithm has been prototyped~\cite{DAT} and currently experimented on programs 
written in an object-oriented language with futures~\cite{ABS}.
The current prototype deals with nonlinear recursive programs by 
using a source-to-source transformation into linear ones. This transformation
\emph{may introduce 
fake dependencies} (which in turn may produce false positives in terms of 
circularities). To briefly illustrate the technique, consider the program

\smallskip

$\qquad\bigl( \, \PP(t) =  (t,x) \sparop (t,y) \sparop \PP(x) \sparop \PP(y)\; , \; 
\PP(u) \; \bigr)$,

\smallskip

\noindent 
Our transformation returns the linear recursive one:

\smallskip

$\begin{array}{l}
\bigl( \, 
\PP^{\it aux}(t,t') =  (t,x)\sparop (t,x') \sparop (t',x)\sparop
(t',x') \sparop \PP^{\it aux}(x,x')  \; ,
\\
\; \; \PP(u) = \PP^{\it aux}(u,u)   \; ,
\; \; \PP(u) \; \; \bigr) \; . 
\end{array}$

\smallskip

To highlight the fake dependencies added by $\PP^{\it aux}$,
we notice that,
after two unfoldings,
$\PP^{\it aux}(u,u)$  gives

\smallskip

$\begin{array}{l}
(u,v) \sparop (u,w) \sparop  (v,v')\sparop
(v,w')\sparop (w,v')\sparop
(w,w') 
\sparop \PP^{\it aux}(v',w') 
\end{array}$

\smallskip

\noindent while $\PP(u)$ has a corresponding state (obtained after four steps)

\smallskip

$\begin{array}{l}
(u,v) \sparop (u,w) \sparop  (v,v') \sparop (v,v'')\sparop  (w,w') \sparop (w,w'')
\\
\sparop \PP(v') \sparop \PP(v'') 
\sparop
\PP(w') \sparop \PP(w'') \; ,
\end{array}
$

\smallskip

\noindent
and this state has no dependency between names created by different invocations. 
It is worth to remark that these additional dependencies cannot be completely
eliminated because of a cardinality argument. The evaluation of a function invocation 
$\M(\wt{u})$ in a linear recursive program may produce at most one invocation of $\M$, 
while an invocation of $\M(\wt{u})$ in a nonlinear recursive program may produce two or
more. In turn, these invocations of $\M$ may create names (which are exponentially 
many in a nonlinear program).
When this happens, the creations of different invocations must be
\emph{contracted} to names created by
one invocation and explicit dependencies must be \emph{added} to account for
dependencies of each invocation.
[Our source-to-source transformation is sound: if the transformed
linear recursive program is circularity-free then the original nonlinear one is also 
circularity-free. 
So, for example,
since our analysis 
lets us determine that the saturated state of
$\PP^{\it aux}$ is circularity-free, then we are able to infer the same property
for $\PP$.]
We are exploring possible generalizations of our theory in 
Section~\ref{sec.algorithm} to nonlinear recursive programs that replace the
notion of mutation with that of \emph{group of mutations}. 
This research direction is currently at an early stage. 

Another obvious research direction is to apply our technique 
to deadlocks due to process synchronizations, as those in 
process calculi~\cite{PIC,Kobayashi06}. In this case, one may take advantage of
 Kobayashi's
inference for deriving inter-channel dependency informations and manage 
recursive behaviors by using our algorithm (instead of the one in~\cite{Typicaltool}).

There are several ways to develop the ideas here, both in terms of the language 
features of lams and the analyses addressed. As regards the lam language, 
\cite{GL2013a} already contains an extension of lams with union types to deal with assignments, data structures, and conditionals.
However, the extension of the theory of mutations and flashbacks to deal with 
these features is not trivial and may yield a weakening of 
Theorem~\ref{thm.invariance}.
Concerning the analyses, the theory of mutations and flashbacks 
may be
applied for verifying properties different than deadlocks, such as
state reachability or livelocks, 
possibly using different lam languages and different notions of
saturated state. Investigating the range of applications
of our theory and studying the related models (corresponding to lams) are two
issues that we intend to pursue.

\bibliographystyle{abbrv}

\bibliography{mfd}

\newpage

\appendix

\section{Java code of the factorial function}
\label{sec.Java}
There are several 
{\sc Java} programs implementing {\tt factorial} in Section~\ref{sec.introduction}. However our goal is to convey some intuition about the differences between
{\sc TyPiCal} and our technique, rather than to analyze the possible options. One option 
is the code

 \begin{lstlisting}[numbers=none]
synchronized void fact(final int n, final int m, final Maths x) 
                                    throws InterruptedException {
        if (n==0) x.retresult(m) ;
        else {
            final Maths y = new Maths() ;
            Thread t = new Thread(new Runnable() {
                public void run() {
                    try { y.fact(n-1,n*m,x) ;
                    } catch (InterruptedException e) { }
                }   }) ;
            t.start();
            t.join() ;
        }
    }
 \end{lstlisting}
Since {\tt factorial} is {\tt synchronized}, the corresponding thread
acquires the lock of its object -- let it be ${\it this}$ -- before
execution and releases the lock upon termination. We notice that 
{\tt factorial}, in case {\tt n>0}, delegates the computation of factorial to
a separate thread on a new object of {\tt Maths}, called $y$. This means that no other 
synchronized thread on ${\it this}$ may be scheduled until the
recursive invocation on $y$ terminates. Said formally, the runtime Java configuration
contains an object dependency $({\it this}, y)$. Repeating this argument for the
recursive invocation, we get configurations with 
chains of dependencies $({\it this}, y), (y,z), \cdots$,
which are finite by the well-foundedness of naturals.

\section{Proof of Theorem~\ref{thm.invariance}.}
\label{sec.thetheory-nfp}

This section develops the technical details for proving Theorem~\ref{thm.invariance}.

\begin{definition}
A history $\alpha$ is 
\begin{description}
\item[$\M$-\emph{yielding}] 
\ \\
if
$\alpha = \alpha_1^{h_1}\beta_1 \cdots 
\alpha_n^{h_n} \beta_n$ such that, for every $i$, 
$\alpha_i$ is a recursive history, $\beta_i \preceq \alpha_i$, 
and $\alpha = \alpha' \M_i$ implies the program has the definition
$\M_i(\wt{x}_i) = \CP[\M(\wt{u}) ]$, for some $\wt{u}$.
The \emph{kernel} of $\alpha$, denoted $[\alpha]$, is 
$\alpha_1^{h_1'}\beta_1 \cdots \alpha_{n}^{h_n'}\beta_{n}$,
where $h_i' = {\tt min}(h_i, 1)$.
\end{description}
\end{definition}

By definition, if $\alpha$ is $\M$-saturating then it is also 
$\M$-yielding. In this case, the kernel $[\alpha]$ has a suffix that
 is $\M$-complete. 
In the $\M\Q\PP$-program,  $\ordermut{\M}=4$,
$\ordermut{\Q}=3$, and $\ordermut{\PP}=1$, and
the recursive histories of $\M$, $\Q$ and $\PP$ are equal to $\M\Q$, to 
$\Q\M$ and to $\PP$, respectively.
Then $\alpha= (\M\Q)^{4}$ is the $\M$-complete history and 
$\alpha'= \PP^2\M$ is $\Q$-yielding, with $[\alpha']=\PP\M$.

We notice that every history of an informative
lam (obtained by evaluating $\ilam{\por{I}_{\P},\,\varnothing ,\, 
\addh{\varepsilon}{\P}}{}$) is a yielding sequence.
We also notice that, for every $\M$, $\varepsilon$ is $\M$-yielding.
In fact, $\varepsilon$ is the history of every function invocation in the
initial lam, which may concern every function name of the program.
As regards the kernel, 
in Lemma~\ref{lem.flashback}, we 
demonstrate that, 
if $\alpha = \alpha_1^{h_1}\beta_1 \cdots %
\alpha_n^{h_n} \beta_n$ is a $\M$-yielding history such that every $h_i \geq 2$,
then every term 
$\suhs{\alpha}{\M(\wt{u})}$ may be mapped by a flashback $\rho$ to a term
$\suhs{[\alpha]}{\M(\rho(\wt{u}))}$; similarly for dependencies. This is the
basic property that allows us to map circularities to past circularities
 (see Theorem~\ref{thm.invariance}). 

Next we introduce an ordering relation over renamings, (in particular, flashbacks) 
and the operation of renaming composition. The definitions are almost standard:
\begin{itemize}
\item 
$\rho \lessfb \rho'$ if, for every $x \in \dom{\rho}$, $\rho(x) = \rho'(x)$.
\item
$\rho  {\scriptstyle \, \circ \,} \rho'$ be defined as follows: 

\smallskip

$\qquad(\rho {\scriptstyle \, \circ \,} \rho')(x) \eqdef 
\left\{ \begin{array}{l@{\qquad}l}
		\rho'(x) & {\it if} \;  \rho'(x) \notin \dom{\rho}
		\\ 
		\rho(\rho'(x)) & {\it otherwise} 
		\end{array} \right.$
\end{itemize}
We notice that, if both
\begin{enumerate}
\item
 $\rho$ and $\rho'$ are flashbacks and 
\item
for every $x \in \dom{\rho}$,  
$\rho'(x) = x$
\end{enumerate}
then $\rho \lessfb \rho {\scriptstyle \, \circ \,} \rho'$ holds.
In the following, lams $\flatt{\P}$ and $\flatt{\mathbb{L}}$, being $\seq$ of
terms that are dependencies composed with $\sparop$, will be written 
$\T_1 \seq \cdots \seq \T_m$ and
$\suh{\T_1} \seq \cdots \seq \suh{\T_m}$, for some $m$, 
respectively, where $\T_i$ and $\suh{\T_i}$ contain dependencies $(x,y)$ and 
$\suhs{\alpha}{(x,y)}$. Let also $\rho(\prod_{i \in I} (x_i,y_i)) = 
\prod_{i \in I}(\rho(x_i), \rho(y_i))$.

With an abuse of notation, we will use the set operation ``$\in$''
for $\P$ and $\suh{\P}$. For instance, we will write $\P' \in \P$  when there is
$\CP[~]$ such that $\P = 
\CP[\P']$. Similarly, we will write $\T \in \T_1 \seq \cdots \seq \T_n$ when there
is $\T_i$ such that $\T \in \T_i$.

A consequence of the
axiom $\T \sparop (\P' \seq \P'')  = \T\sparop \P' \seq \T \sparop \P''$ is 
the following property of the informative operational semantics.

\begin{proposition}
\label{prop.storie-insieme}
Let $\ilam{\por{V}_1,\suh{\Mplus}, \suh{\CP}_0[\suhs{\alpha}{\M_1(\wt{u_1})]}}{}$
be a state of an informative operational semantics. For every $1 \leq i 
\leq n$, let $\M_i(\wt{u_i}) = \P_i'$ and $\addh{\alpha\M_0 \cdots \M_i}{\P_i'}$ be 
$\suh{\CP}_i[\suhs{\alpha\M_1 \cdots \M_i}{\M_{i+1}(\wt{u_{i+1}})}]$. 
Finally, let

\smallskip

$\begin{array}{rl}
\flatt{\suh{\CP}_1[\cdots \suh{\CP}_n[\suhs{\alpha\M_1 \cdots \M_n}{\M_{n+1}(\wt{u_{n+1}})}] \cdots]} = &
\suh{\T}_1 \seq \cdots \seq \suh{\T}_r
\\
\flatt{\suh{\CP}_n[\suhs{\alpha\M_1 \cdots \M_n}{\M_{n+1}(\wt{u_{n+1}})}]} = 
& \suh{\T}_1' \seq \cdots \seq \suh{\T}_{r'}'\; .
\end{array}$

\smallskip

If  \ $\suhs{\alpha\M_1 \cdots \M_n}{(x,y)} \sparop 
\addh{\alpha'}{\T} \in \suh{\T}_1 \seq \cdots \seq \suh{\T}_r$ 
then, for every $1 \leq j \leq r'$, 
$\suh{\T}_{j}' \sparop \addh{\alpha'}{\T} \in \suh{\T}_1 \seq \cdots \seq \suh{\T}_r$.
\end{proposition}

The next lemma allows us to map, through a flashback, terms in a saturated 
state to terms that have been produced in the past. The correspondence is defined
by means of the (regular) structure of histories. 

\begin{lemma}
\label{lem.flashback}
Let $\ilam{\por{I}_{\P},\,\varnothing ,\, \addh{\varepsilon}{\P}
}{} \lred{}^* \ilam{\mathbb{V},\,\suh{\Mplus}, \, \mathbb{L}}{}$ 
and $\ilam{\mathbb{V}, \,\suh{\Mplus} ,\, \mathbb{L}}{}$ 
be saturated and 
$\flatt{\mathbb{L}}=\suh{\T}_1\seq\cdots\seq\suh{\T}_m$. Then
\begin{enumerate}
\item
if $\suhs{\beta\alpha^{n+2}\beta'}{\M(\wt{u})}\in\mathbb{L}$, where 
$\beta\alpha^{n+2}\beta'$ is $\M$-yielding, then there
are $n+1$ $\mathbb{V}$-flashbacks $\rho_{\beta,\alpha,\beta'}^{(2)}, \cdots  
, \rho_{\beta,\alpha,\beta'}^{(n+2)}$ such that:
\begin{enumerate}
\item
$\suhs{\beta \alpha^{n+1}\beta'}{\M(\rho_{\beta,\alpha,\beta'}^{(n+2)}(\wt{u}))} \in 
\suh{\Mplus}$;
\item
$\prod_{j\in J}\addh{\beta \alpha^{k+1}\beta_j}{\T_j'} \in \suh{\T}_1\seq\cdots\seq
\suh{\T}_m$ where, 
for every $j$, $\beta_j \preceq \alpha$, implies 
$\prod_{j\in J}\addh{\beta \alpha^{k}\beta_j}{\rho_{\beta,\alpha,\beta'}^{(k+1)}(\T_j')} 
\in \suh{\T}_1\seq\cdots\seq\suh{\T}_m$;
\item
$\suhs{\beta \alpha^{k+1}\beta'}{\M(\wt{u})} \in \suh{\Mplus}$ 
 implies 
$\suhs{\beta \alpha^{k}\beta'}{\M(\rho_{\beta,\alpha,\beta'}^{(k+1)}(\wt{u}))} \in 
\suh{\Mplus}$.
\end{enumerate}

\item 
if $\alpha_1, \cdots, \alpha_k$ are  %
 $\M_1$-yielding, $\cdots$, $\M_k$-yielding, respectively,
  then there are flashbacks $\rho_{\alpha_1}, \cdots, \rho_{\alpha_k}$ 
such that
\begin{enumerate}
\item
if $\suhs{\alpha_1}{\M_1(\wt{u})} \in \mathbb{L}$ or 
$\suhs{\alpha_1}{\M_1(\wt{u})} \in  \suh{\Mplus}$ 
then $\suhs{[\alpha_1]}{\M(\rho_{\alpha_1}(\wt{u}))} \in \suh{\Mplus}$;
\item
if $\prod_{1 \leq j \leq k}\addh{\alpha_j}{\T_j} \in \suh{\T}_1\seq\cdots\seq\suh{\T}_m$ 
then \\
$\prod_{1 \leq j \leq k}\addh{[\alpha_j]}{\rho_{\alpha_j}(\T)} \in \suh{\T}_1\seq\cdots\seq\suh{\T}_m$;

\item
if $\alpha_1 \preceq \alpha_2$ then $\rho_{\alpha_1} \lessfb \rho_{\alpha_2}$.

(In particular, if $\alpha_1 = \beta \alpha^{n+2}\beta'$, with $\beta' \preceq 
\alpha$, and $\alpha_2 = \beta \alpha^{n+3}$ then
$\rho_{\alpha_1} \lessfb \rho_{\alpha_2}$).
\end{enumerate}
\end{enumerate}
\end{lemma}

\begin{proof} (Sketch)
As regards item 1, let  $\alpha = \beta'\beta''$ and let $\beta'' \beta' = \M \M_1 \cdots \M_m$ (therefore the length of $\alpha$ is $m+1$).
The evaluation $\ilam{\por{I}_{\P},\,\varnothing ,\, \addh{\varepsilon}{\P}
}{} \lred{}^* \ilam{\mathbb{V},\,\suh{\Mplus}, \, \mathbb{L}}{}$ may be decomposed as follows
\[
\begin{array}{rl}
\ilam{\por{I}_{\P},\,\varnothing ,\, \addh{\varepsilon}{\P}
}{} \lred{}^* &
\ilam{\mathbb{V}', \, \suh{\Mplus}', \, 
\suh{\CP}[\suhs{\beta\alpha^{n+1}\beta'}{\M(\wt{u'})}]}{}
\\
\qquad \lred{}^{*} &
\ilam{\mathbb{V}, \, \suh{\Mplus}, \, \mathbb{L}}{}
\end{array}
\] 
By definition of the operational semantics there is the \emph{alternative} evaluation
\[
\begin{array}{@{\!\!\!}l}
\ilam{\mathbb{V}', \, \suh{\Mplus}', \, 
\suh{\CP}[\suhs{\beta\alpha^{n+1}\beta'}{\M(\wt{u'})}]}{}
\\
\lred{} \ilam{\mathbb{V}'', \, \suh{\Mplus}'', \, 
\suh{\CP}[\suh{\CP}'[\suhs{\beta\alpha^{n+1}\beta'\M}{\M_1(\wt{u_{1}})}] ]}{}
 \\
\! \! \lred{}^{*}
\ilam{\mathbb{V}''', \, \suh{\Mplus}''', \, 
\suh{\CP}[ \suh{\CP}'[\suh{\CP}_1[ \cdots \suh{\CP}_m[
\suhs{\beta\alpha^{n+1}\beta'\M\M_1 \cdots \M_m}{\M(\wt{u''})}] \cdots]]]}{}
\end{array}
\] 
[notice that $\beta\alpha^{n+1}\beta'\M\M_1 \cdots \M_m = \beta\alpha^{n+2}\beta'$].
Property (1.a) is an immediate consequence of Proposition~\ref{prop.flashback-order};
let $\varrho_{\beta,\alpha,\beta'}^{(n+2)}$ be the flashback 
for the last state. The property (1.b), when $k=n$, is also an immediate consequence of 
Propositions~\ref{prop.flashback-order}
and of~\ref{prop.storie-insieme}. In the general case, we need to iterate 
the arguments on shorter histories and the arguments are similar for (1.c).
In order to conclude the proof of item 1, we need an additional argument. By Proposition~\ref{prop.diamond}, there exists an evaluation
\[
\begin{array}{l}
\ilam{\mathbb{V}''', \, \suh{\Mplus}''', \, 
\suh{\CP}[\suh{\CP}'[\suh{\CP}_1[ \cdots \suh{\CP}_m[
\suhs{\beta\alpha^{n+1}\beta'\M\M_1 \cdots \M_m}{\M(\wt{u''})}] \cdots]]]}{}
\\
\qquad \lred{}^{*} \ilam{\mathbb{V}^\sharp, \, \suh{\Mplus}^\sharp, \, 
\mathbb{L}^\sharp}{}
\end{array}
\]
such that $\ilam{\mathbb{V}^\sharp, \, \suh{\Mplus}^\sharp, \, 
\mathbb{L}^\sharp}{}$ and $\ilam{\mathbb{V}, \, \suh{\Mplus}, \, 
\mathbb{L}}{}$ are identified by a bijective renaming, let it be $\jmath$. We define
the $\rho_{\beta,\alpha,\beta'}^{(n+2)}$
corresponding to the evaluation
$\ilam{\por{I}_{\P},\,\varnothing ,\, \addh{\varepsilon}{\P}
}{} \lred{}^* \ilam{\mathbb{V},\,\suh{\Mplus}, \, \mathbb{L}}{}$
as $\rho_{\beta,\alpha,\beta'}^{(n+2)} \eqdef \jmath \circ 
\varrho_{\beta,\alpha,\beta'}^{(n+2)} \circ \jmath^{-1}$. Similarly for the other
$\rho_{\beta,\alpha,\beta'}^{(k+1)}$. The properties of item 1 for 
$\ilam{\mathbb{V}, \, \suh{\Mplus}, \, 
\mathbb{L}}{}$ follow
by the corresponding ones for
\[
\ilam{\mathbb{V}''', \, \suh{\Mplus}''', \, 
\suh{\CP}[\suh{\CP}'[\suh{\CP}_1[ \cdots \suh{\CP}_m[
\suhs{\beta\alpha^{n+1}\beta'\M\M_1 \cdots \M_m}{\M(\wt{u''})}] \cdots]]]}{} \; .
\]

We prove item 2. We observe that a term with history
$\beta_0(\alpha_1')^{h_1}$ $\beta_1 \cdots %
\beta_{n-1}
(\alpha_n')^{h_n} \beta_n$ in $\suh{\Mplus}$ or in $\mathbb{L}$ may have no
corresponding term (by a flashback) with history
$\beta_0(\alpha_1')^{h_1-1}\beta_1$ $(\alpha_2')^{h_2} \cdots \beta_{n-1}$ 
$(\alpha_n')^{h_n} \beta_n$. This is because the evaluation to the saturated state
may have not expanded some invocations. It is however true that terms with 
histories $[\beta_0(\alpha_1')^{h_1}\beta_1 \cdots 
\beta_{n-1}(\alpha_n')^{h_n} \beta_n]$ (kernels) are either in $\suh{\Mplus}$ or in 
$\mathbb{L}$ and the item 2 is demonstrated by proving that a flashback to terms
with histories that are kernels does exist. 

Let $\alpha_1 = \beta_0(\alpha_1')^{h_1}\beta_1 \cdots 
\beta_{n-1}(\alpha_n')^{h_n}\beta_n$ be a $\M$-yielding sequence. We proceed by induction on $n$.
When $n = 1$ there are two cases: $h_1 \leq 1$ and $h_1\geq 2$. In the first
case there is nothing to prove because $[\alpha] = \alpha$. When $h_1\geq 2$,
since $\alpha$ fits with the hypotheses of Item 1, there exist
$\rho_{\beta_0,\alpha_1',\beta_1}^{(2)}, \cdots ,\rho_{\beta_0,\alpha_1',\beta_1}^{(h_1)}$.  
Let $\delta_{\beta_0,\alpha_1',\beta_1}^{(2)} = \rho_{\beta_0,\alpha_1',\beta_1}^{(2)}$ and 
$\delta_{\beta_0,\alpha_1',\beta_1}^{(i+1)} = \rho_{\beta_0,\alpha_1',\beta_1}^{(i+1)}[x \mapsto x \; | \; 
x \in \dom{\delta_{\beta_0,\alpha_1',\beta_1}^{(i)}}]$.
We also let $\rho_{\alpha_1} =\delta_{\beta_0,\alpha_1',\beta_1}^{(2)} {\scriptstyle \, \circ \,} \cdots {\scriptstyle \, \circ \,} \delta_{\beta_0,\alpha_1',\beta_1}^{(h_1)}$ and we observe that, by definition of renaming composition, 
if $\alpha_1 \preceq \alpha_2$ then
$\rho_{\alpha_1} \lessfb \rho_{\alpha_2}$. 
In this case, the items 2.a and 2.b follow by item 1, 
Proposition~\ref{prop.storie-insieme} and the diamond property of Proposition~\ref{prop.diamond}.

We assume the statement holds for a generic $n$ and we prove the case $n+1$.
Let $\alpha_1 = \beta 
\beta_{n}(\alpha_{n+1}')^{h_{n+1}}\beta_{n+1}$ and $h_{n+1} >0$ (because
$[\beta_{n}(\alpha_{n+1}')^1\beta_{n+1}] = \beta_{n}\alpha_{n+1}'\beta_{n+1}$). We consider the map
\[
\rho_{\alpha_1} \eqdef \rho_{\beta} {\scriptstyle \, \circ \,} \delta_{\beta_n,\alpha_{n+1}',\beta_{n+1}}^{(2)} {\scriptstyle \, \circ \,}
\cdots {\scriptstyle \, \circ \,} \delta_{\beta_n,\alpha_{n+1}',\beta_{n+1}}^{(h_{n+1})}\]
where $\delta_{\beta_n,\alpha_{n+1}',\beta_{n+1}}^{(i)}$, $2 \leq i \leq h_{n+1}$
are defined as above.
As before, the items 2.a and 2.b follow by item 1 for 
$\delta_{\beta_n,\alpha_{n+1}',\beta_{n+1}}^{(2)} {\scriptstyle \, \circ \,}
\cdots$ ${\scriptstyle \, \circ \,} \delta_{\beta_n,\alpha_{n+1}',\beta_{n+1}}^{(h_{n+1})}$ and by
Proposition~\ref{prop.storie-insieme} and the diamond property of Proposition~\ref{prop.diamond}.  Then we apply the inductive hypothesis for $\rho_{\beta}$.
The property (2.c) $\alpha_1 \preceq \alpha_2$ implies
 $\rho_{\alpha_1} \lessfb \rho_{\alpha_2}$ is an immediate consequence of the definition.
\end{proof}

Every preliminary statement is in place for our key theorem that details the mapping
of circularities created by transitions of saturated states to past circularities.
For readability sake, we restate the theorem.

\bigskip

\noindent
{\bf Theorem~\ref{thm.invariance}.} {\em
Let $\ilam{\por{I}_{\P}, \, \varnothing, \,  \addh{\varepsilon}{\P}}{} \lred{}^* \ilam{\mathbb{V}, \, \suh{\Mplus}, \, \mathbb{L}}{}$ 
and
$\ilam{\mathbb{V}, \,  \suh{\Mplus}, \,\mathbb{L}}{}$ be a saturated state.
If
$\ilam{\mathbb{V},  \suh{\Mplus}, \,\mathbb{L}}{} \lred{}
\ilam{\mathbb{V}', \,  \suh{\Mplus}', \,\mathbb{L}'}{}$ then
\begin{enumerate}
\item
$\ilam{\mathbb{V}', \,  \suh{\Mplus}', \,\mathbb{L}'}{}$ is saturated;
\item
if $\mathbb{L}'$ has a circularity then $\mathbb{L}$ has already 
a circularity.
\end{enumerate}
}

\begin{proof}
The item 1. is an immediate consequence of Proposition~\ref{prop.flashback-order}.
We prove 2.  
Let
\begin{itemize}
\item[--]
$\mathbb{L} =  \suh{\CP}[\suhs{\alpha}{\M(\wt{u})} ]$;
 
\item[--]
$\M(\wt{u}) = \P'$

\item[--]
$\mathbb{L}' = \suh{\CP}[\addh{\alpha\M}{\P'}]$; %

\item[--]
$\flatt{\mathbb{L}}=\flatt{\suh{\CP}[\suhs{\alpha}{\M(\wt{u})} ]}=\suh{\T}_1\seq\cdots\seq\suh{\T}_p$;

\item[--]
$\flatt{\P'} = \T'_1\seq \cdots \seq \T_{p'}'$;

\item[--]
$\flatt{\mathbb{L}'}
= \suh{\T}_1''\seq\cdots\seq\suh{\T}_q''$;

\item[--]
$\suhs{\alpha_0}{(x_0,x_1)} \sparop \cdots \sparop
\suhs{\alpha_n}{(x_n,x_0)} \in \suh{\T}_1''\seq\cdots\seq\suh{\T}_q''$ 
 (it is a circularity).
\end{itemize}
Without loss of generality, we may reduce to the following case (the general case
is demonstrated by iterating the arguments below).

Let $\alpha \M = \beta (\alpha')^{m+2}\beta'$ and let 
\[
\begin{array}{@{\!}rl}
\suhs{\alpha_0}{(x_0,x_1)} \sparop \cdots \sparop
\suhs{\alpha_n}{(x_n,x_0)} =  &
\prod_{0 \leq j \leq n'} \suhs{\beta(\alpha')^{m+1}\beta'\beta_j}{(x_j, x_{j+1})}
\\
& \sparop \suhs{\alpha_{n'+1}}{(x_{n'+1}, x_{n'+2})}
\\
& \sparop \cdots \sparop
\suhs{\alpha_n}{(x_n,x_0)}
\end{array}
\] 
with $\varepsilon \precneq \beta_j \preceq \beta''\beta'$, where $\beta'\beta'' = \alpha'$, %
and $n'<n$
(otherwise 2 is straightforward because the circularity may be mapped
to a previous circularity by $\rho_{\beta,\alpha',\beta'}^{(m+2)}$, see Lemma~\ref{lem.flashback}(1.b), or it
is already contained in $\mathbb{L}$). This is the case of crossover circularities, as discussed in Section~\ref{sec.introduction}.

By Lemma~\ref{lem.flashback}, 
\begin{eqnarray}
\label{eq.lemma}
\begin{array}{l}
\suhs{\beta (\alpha')^{m}\beta'\beta_0}{(\rho_{\beta,\alpha,\beta'}^{(m+2)}(x_0), 
\rho_{\beta,\alpha,\beta'}^{(m+2)}(x_{1}))} \sparop  \cdots  
\\
\sparop \;
\suhs{\beta (\alpha')^{m+1}\beta'\beta_{n'}}{
(\rho_{\beta,\alpha,\beta'}^{(m+2)}(x_{n'}), 
\rho_{\beta,\alpha,\beta'}^{(m+2)}(x_{n'+1}))}
\end{array}
\end{eqnarray}
is in some $\suh{\T}_{i}''$. There are two cases.

\emph{Case 1}: for every $n'+1 \leq i \leq n$, 
$\alpha_{i} \precneq \beta (\alpha')^{m+1}\beta'$. Then, by 
Lemma~\ref{lem.flashback}(1), we have $\rho_{\beta,\alpha,\beta'}^{(m+2)}(x_{0}) = 
\rho_{\beta,\alpha,\beta'}^{(m+1)}(x_{0})$ and 
$\rho_{\beta,\alpha,\beta'}^{(m+2)}(x_{n'+1}) = 
\rho_{\beta,\alpha,\beta'}^{(m+1)}(x_{n'+1})$. Therefore, by 
Lemma~\ref{lem.flashback}(2), 
\[
\begin{array}{rl}
(\ref{eq.lemma})
\sparop &
\suhs{\alpha_{n'+1}'}{(\rho_{\beta,\alpha,\beta'}^{(m+1)}(x_{n'+1}),\rho_{\beta,\alpha,\beta'}^{(m+1)}(x_{n'+2}))}
\\
\sparop & \cdots \sparop \suhs{\alpha_{n}'}{(
\rho_{\beta,\alpha,\beta'}^{(m+1)}(x_{n}),\rho_{\beta,\alpha,\beta'}^{(m+1)}(x_{0}))}
\end{array}
\] 
with suitable $\alpha_{n'+1}', \cdots , \alpha_n'$, 
is a circularity in $ \suh{\T}_1''\seq\cdots\seq\suh{\T}_q''$. In particular, 
whenever, for every $n'+1 \leq i \leq n$, $\alpha_i = \beta (\alpha')^{m}\beta'\beta_i$ 
with $\varepsilon \precneq \beta_{i} \preceq \beta''\beta'$, the flashback 
$\rho_{\beta,\alpha,\beta'}^{(m+1)}$ maps dependencies $\suhs{\alpha_i}{(x_i,x_{i+1})}$ to 
dependencies 
\[
\suhs{\beta (\alpha')^{m-1}\beta'\beta_i}{(\rho_{\beta,\alpha,\beta'}^{(m+1)}(x_i),\rho_{\beta,\alpha,\beta'}^{(m+1)}(x_{i+1}))}
\]
 if $m >0$. It is the identity,
if $m=0$.

\emph{Case 2}: there is $n'+1 \leq i \leq n$ such that $\alpha_i \not \preceq \beta (\alpha')^{m+2}\beta'$. Let this $i$ be $n'+1$. For instance, 
$\beta = \beta_1' (\alpha'')^{m'} \beta_1''$ and $\alpha_{n'+1} = \beta_1'(\alpha'')^{m'+1} \beta_1''(\alpha''')^{m''} \beta_1'''$ with $m' \geq 2$ and $m'' \geq 2$.
In this case it is possible that there is no pair $\suhs{\gamma}{(y, 
y')}$, with $\gamma \succeq  \beta_1' (\alpha'')^{m'}$, to which map 
$\suhs{\alpha_{n'+1}}{(x_{n'+1}, x_{n'+2})}$ by means of a flashback.
To overcome 
this issue, we consider the flashbacks
$\rho_{\alpha_0}, \cdots , \rho_{\alpha_{n'}}, \rho_{\alpha_{n'+1}}$ and we observe 
that
\begin{eqnarray}
\label{eq.lemma2}
\begin{array}{l}
\suhs{[\alpha_0]}{(\rho_{\alpha_0}(x_0), 
\rho_{\alpha_0}(x_{1}))} \sparop  \cdots  
\sparop \;
\suhs{[\alpha_{n'}]}{
(\rho_{\alpha_{n'}}(x_{n'}), 
\rho_{\alpha_{n'}}(x_{n'+1}))} \hspace{-.6cm}
\\
\sparop \; \suhs{[\alpha_{n'+1}]}{
(\rho_{\alpha_{n'+1}}(x_{n'+1}), 
\rho_{\alpha_{n'+1}}(x_{n'+2}))} \sparop \cdots 
\\
\sparop \; \suhs{[\alpha_{n}]}{(\rho_{\alpha_{n}}(x_{n}), 
\rho_{\alpha_{n}}(x_{1}))}
\end{array}
\end{eqnarray}
verifies
\begin{enumerate}
\item[(a)]
for every $0 \leq i < n$, $\rho_{\alpha_i}(x_{i+1}) = 
\rho_{\alpha_{i+1}}(x_{i+1})$ and $\rho_{\alpha_n}(x_{0}) = 
\rho_{\alpha_{0}}(x_{0})$;
\item[(b)]
the term (\ref{eq.lemma2}) is a subterm of $ \suh{\T}_1''\seq\cdots\seq\suh{\T}_q''$.
\end{enumerate}
As regards (a), the property derives by definition of the flashbacks $\rho_{\alpha_i}$
and $\rho_{\alpha_{i+1}}$ in Lemma~\ref{lem.flashback}.
As regards (b), it follows by Lemma~\ref{lem.flashback}(2.b) because 
$\suhs{\alpha_0}{(x_0,x_1)} \sparop \cdots \sparop$
$\suhs{\alpha_{n}}{(x_{n},x_{1})} \in \suh{\T}_1''$ $\seq\cdots\seq\suh{\T}_q''$.
\end{proof}

\section{Nonlinear programs: technical aspects}
\label{sec.nonlinear}

When the lam program is not linear recursive, it is not possible to associate a unique
mutation to a function. In the general case, our technique for verifying 
circularity-freedom consists of transforming a nonlinear recursive program
into a linear recursive one and then running the algorithm of the previous section. As we will see, the transformation
 introduces inaccuracies, e.g.~dependencies that are not
 present in the nonlinear recursive program.

\subsection{The pseudo-linear case} 
\label{sec.pseudolinear}
In nonlinear recursive programs, recursive histories
are no more adequate to capture the mutations defined by the functions.
For example, in the nonlinear recursive program (called $\M'\N'$-program)

\smallskip

$\! \! \!\!\!\!\!\!\!\!\bigl( \M'(x,y,z) = (x,y) \sparop \N'(y,z), \;
\N'(x,y) = \N'(x,z) \sparop \M'(z,y,y), \; \P \bigr)$ 

\smallskip

\noindent the recursive
history of $\M'$ is $\M'\N'$.
The sequence $\M'\N'\N'$ \emph{is not} a recursive history
because it contains multiple occurrences of the function $\N'$.
However, if one computes the sequences of invocations 
$\M'(x,y,z) \cdots \M'(\wt{u})$,
it is possible to derive the two sequences $\M'(x,y,z) \N'(y,z) \M'(z',z,z)$
and $\M'(x,y,z) \N'(y,z)$ $\N'(y,u) \M'(u',u,u)$ that define  two different
mutations $\mut{4,3,3}$ and $\mut{6,5,5}$ (see the definition of mutation of a
function). 

\begin{definition}
A program $\bigl(\M_1(\wt{x_1}) = \P_{1}, \cdots , 
\M_\ell(\wt{x_\ell}) = \P_{\ell}, \P \bigr)$ is \emph{pseudo-linear recursive} if, for every
$\M_{i}$, the set of functions $\{ \M \; | \; \closure{\M} = \closure{\M_i} \}$
contains at most one function with a number of recursive histories greater than 1.
\end{definition}
The $\M'\N'$-program above is pseudo-linear recursive, as well as the ${\tt fibonacci}$
program in Section~\ref{sec.introduction} and the following $\Q'$-program

\smallskip

$\bigl( \Q'(x,y,z) = (x,y) \sparop \Q'(y,z,x) \seq (x,u) \sparop \Q'(u,u,y)
, \; \P \bigr)$ .

\smallskip

In these cases, functions have a unique recursive history
but there are multiple recursive invocations. 
On the contrary, the  $\M''\N''$-program below

\smallskip
 
$\begin{array}{ll}
\bigl( \; & \M''(x,y) =  (x,z)\sparop\M''(y,z) \seq \N''(y,x) \; ,
\\
 &\N''(x,y) = (y,x)\sparop\M''(y,z) \sparop \N''(z,x) \;,
 \\
 & \M''(x_1,x_2) \quad
 \bigr)
\end{array}$

\smallskip

\noindent is not pseudo-linear recursive.

Pseudo-linearity has been introduced because of the easiness of transforming them
into linear recursive  programs. The transformation consists of the three steps 
specified in Table~\ref{tab.trans}, which we discuss below.
 Let $\bigl( \M_1(\wt{x_1}) = \P_1, \cdots ,  \M_\ell(\wt{x_\ell}) = 
 \P_\ell, \P \bigr)$ be a lam program, let $\rechis{\M_i}$ be the set of recursive histories of $\f_i$, and let $\head{\varepsilon} = \varepsilon$ and 
 $\head{\M\alpha} = \M$.

\begin{table*}[t]
{\footnotesize
\[
\begin{array}{c}
\bigfract{
	\begin{array}{c}
	\rechis{\M_i} = \{ \M_i \M_k \alpha , \M_i \beta_0 , \cdots , \M_i \beta_n\}
	\qquad  \{ \head{\beta_0} , \cdots , \head{\beta_n} \}\setminus \M_k
	 \neq \varnothing
	\\
	\P_i = \CP[\M_k(\wt{u})] \qquad
	\var{\P_k}\setminus\wt{x_k}=\wt{z} \qquad \wt{w} \mbox{ are fresh}
	 \end{array}
}{
\transfuno{ \bigl( \cdots \M_i(\wt{x_i}) = \P_i, \cdots , \P \bigr)
}{ \bigl( \cdots \M_i(\wt{x_i}) = \CP[\P_k\subst{\wt{w}}{\wt{z}}\subst{\wt{u}}{\wt{x_i}}], \cdots , \P \bigr) }
}
\\
\\
\bigfract{
	\begin{array}{c}
	\rechis{\M_i} = \{\M_i \alpha \} \qquad \M_k = \head{\alpha}
	\\
	\P_i = \CP[\M_k(\wt{u_0})] \cdots [\M_k(\wt{u_{n+1}})] \qquad \M_k \notin \CP
	\\
	\var{\P_k}\setminus\wt{x_k}=\wt{z} 
	\qquad
	\wt{w_0}, \cdots , \wt{w_{n+1}} \mbox{ are fresh}
	\\
	 \CP[\P_k\subst{\wt{w_0}}{\wt{z}}\subst{\wt{u_0}}{\wt{x_k}}] \cdots  
	 [\P_k\subst{\wt{w_{n+1}}}{\wt{z}}\subst{\wt{u_{n+1}}}{\wt{x_k}}]= \P_i'
	 \end{array}
}{
\transfdue{ \bigl( \cdots \M_i(\wt{x_i}) = \P_i, \cdots , \P \bigr)
}{ \bigl( \cdots \M_i(\wt{x_i}) = \P_i', \cdots , \P \bigr) }
}
\\
\\ 
\bigfract{
	\begin{array}{c}
	\P_i = \CP[\M_i(\wt{u_0})] \cdots [\M_i(\wt{u_{n+1}})] \qquad
	\M_i \notin \CP
	\qquad
	\wt{w_0}, \cdots , \wt{w_{n+1}} \mbox{ are fresh}
	\\
	 \P_i^{\it aux} = \M_i^{\it aux}(\wt{u_0}\subst{\wt{w_0}}{\wt{x_i}}, \cdots , 
	 \wt{u_{n+1}}
\subst{\wt{w_{n+1}}}{\wt{x_i}}) \sparop 
(\prod_{j\in 0..n+1} \flatm{\M_i}{\P_i} \subst{\wt{w_j}}{\wt{x_i}}) 
	 \end{array}
}{
\begin{array}{l}
\bigl( \cdots \M_i(\wt{x_i}) = \P_i, \cdots , \P \bigr)
\; \stackrel{{\tt pl} \mapsto {\tt l}}{\Longmapsto}_3
\\
\quad  \bigl( \cdots \M_i(\wt{x_i}) = \M_i^{\it aux}(\underbrace{\wt{x_i}, \cdots , \wt{x_i}}_{n+2 \; {\rm times}}), \; 
\M_i^{\it aux}(\wt{w_0}, \cdots, \wt{w_{n+1}}) = \P_i^{\it aux}, \cdots , \P \bigr) 
\end{array}
}
\end{array}
\]}
\caption{\label{tab.trans} Pseudo-linear to 
linear transformation
}
\end{table*}

\medskip

\noindent
\emph{Transformation $\stackrel{{\tt pl} \mapsto {\tt l}}{\Longmapsto}_1$: Removing multiple recursive histories}.
We repeatedly apply the rule defining 
$\stackrel{{\tt pl} \mapsto {\tt l}}{\Longmapsto}_1$.
Every instance of the rule selects a function $\M_i$ with a number of recursive histories 
greater than one -- the hypotheses 
$\rechis{\M_i} = \{ \M_i \M_k \alpha , \M_i \beta_0 , \cdots , \M_i \beta_n\}
$ and $\{ \head{\beta_0} , \cdots , \head{\beta_n} \}\setminus \M_k
\neq \varnothing$ -- and expands the invocation of $\M_k$, with $\M_k \neq \M_i$. 
By definition of
pseudo-linearity, the other function names in $\rechis{\M_i}$ have
one recursive history. At each application of the rule %
the sum of the lengths of the recursive histories of $\M_i$ decreases.
Therefore we eventually unfold the (mutual) recursive invocations
of $\M_i$ till the recursive history of $\M_i$ is unique.
For example, the program

\smallskip

$\bigl( \M(x) = (x,y) \sparop \N(x) \; , \; \N(x) = (x,y) \sparop \M(x) \seq \N(y) \; ,   \;\P \; \bigr)$

\smallskip

\noindent is transformed into

\smallskip

$\!\!\!\!\!\!\bigl( \M(x) = (x,y) \sparop \N(x) \; , \; \N(x) = (x,y) \sparop (x,z) \sparop \N(x) \seq \N(y) \; ,   \;\P \; \bigr).$

\medskip

\noindent
\emph{Transformation $\stackrel{{\tt pl} \mapsto {\tt l}}{\Longmapsto}_2$: 
Reducing the histories of pseudo-linear recursive functions.}
By 
$\stackrel{{\tt pl} \mapsto {\tt l}}{\Longmapsto}_1$, we are reduced to
functions that have one recursive
history. 
Yet, this is not enough for a program to be linear recursive, such as the $\Q'$-program or the following
$\PP''\Q''$-program

\smallskip

$\begin{array}{ll}
\bigl( \; & \PP''(x,y) =  (x,z)\sparop\Q''(y,z) \seq \Q''(y,x) \; ,
\\
 &\Q''(x,y) = (y,x)\sparop\PP''(y,z) \sparop \PP''(z,x) \;,
 \\
 & \PP''(x_1,x_2) \quad
 \bigr)
\end{array}$

\smallskip

\noindent (the reason is that the bodies of functions may have different invocations of a same function).
Rule $\stackrel{{\tt pl} \mapsto {\tt l}}{\Longmapsto}_2$ expands the
bodies of pseudo-linear recursive functions till the histories of nonlinear recursive functions
have length one.
In this rule (and in the following
), we use lam contexts 
with multiple holes, written $\CP[ \, ] \cdots [\, ]$. We write
$\M \notin \CP$%
whenever there is no invocation of $\M$ in $\CP$. 

By the hypotheses of the rule, it applies to a function $\M_i$ whose next element
in the recursive history is $\M_k$ (by definition of the recursive history, $\M_i
\neq \M_k$) and %
whose body $\P_i$ contains \emph{at least} two invocations of 
$\M_k$. The rule transforms  $\P_i$ %
by expanding every invocation of $\M_k$.
For example, the functions $\PP''$ and $\Q''$ in the
$\PP''\Q''$-program are transformed into
 
\smallskip

$\begin{array}{rl}
\PP''(x,y) = &  (x,z)\sparop\Q''(y,z) \seq \Q''(y,x) \; ,
\\
 \Q''(x,y) = & (y,x)\sparop ((y,z')\sparop\Q''(z,z')  \seq \Q''(z,y)) 
 \\
 & \sparop ((z,z'')\sparop\Q''(x,z'') \seq  \Q''(x,z)).
\end{array}$

\smallskip

The arguments about the termination of the transformation 
$\stackrel{{\tt pl} \mapsto {\tt l}}{\Longmapsto}_2$ are straightforward.

\medskip

\noindent
\emph{Transformation $\stackrel{{\tt pl} \mapsto {\tt l}}{\Longmapsto}_3$: 
Removing nonlinear recursive invocations}.
By $\stackrel{{\tt pl} \mapsto {\tt l}}{\Longmapsto}_2$ we are reduced to 
pseudo-linear recursive programs
where the nonlinearity is due to recursive, but not mutually-recursive functions (such
as ${\tt fibonacci}$).
The transformation $\stackrel{{\tt pl} \mapsto {\tt l}}{\Longmapsto}_3$ removes
multiple recursive invocations of nonlinear recursive programs. This transformation 
 is the one that introduces
inaccuracies, e.g. pairs that are not present in the nonlinear recursive program.

In the rule of $\stackrel{{\tt pl} \mapsto {\tt l}}{\Longmapsto}_3$ we use the
auxiliary operator $\flatm{\M}{\P}$ defined as follows:

\smallskip

$\begin{array}{l@{\qquad}l}
\flatm{\M}{\pinull} =  \pinull,
&
\flatm{\M}{(x,y)}  =  (x,y),
\\
\flatm{\M}{\M(\wt{x})}  =  \pinull,
&
\flatm{\M}{\N(\wt{x})}  =  \N(\wt{x}) \mbox{, if }(\M \neq \N),
\\
\flatm{\M}{\P \sparop \P'}  =  \flatm{\M}{\P} \sparop \flatm{\M}{\P'} ,
&
\flatm{\M}{\P \seq \P'}  =  \flatm{\M}{\P} \seq \flatm{\M}{\P'}.
\end{array}$

\smallskip

The rule of $\stackrel{{\tt pl} \mapsto {\tt l}}{\Longmapsto}_3$ selects a
function $\M_i$ whose body contains multiple recursive invocations and extracts all of them 
-- the term $\flatm{\M_i}{\P_i}$. This term is put in parallel with an
auxiliary function invocation -- the function $\M_i^{\it aux}$ -- that collects
the arguments of each invocation $\M_i$ (with names that have been 
properly renamed). The resulting term, called
$\P_i^{\it aux}$ is the body of the new function $\M_i^{\it aux}$ that is invoked
by $\M_i$ in
the transformed program.
For example, the function ${\tt fibonacci}$ 
\[
{\tt fibonacci}(r,s) =  (r,s) \sparop (t,s) \sparop {\tt fibonacci}(r,t) \sparop {\tt fibonacci}(t,s)
\]
is transformed into

\smallskip
 
$\begin{array}{lcl}
{\tt fibonacci}(r,s) & = & {\tt fibonacci}^{\it aux}(r,s,r,s), 
\\
{\tt fibonacci}^{\it aux}(r,s,r',s') & = & (r,s)\sparop
(r',s')
\\
&&
\sparop  {\tt fibonacci}^{\it aux}(r,t,t,s') 
\end{array}$

\smallskip

\noindent where different invocations (${\tt fibonacci}(r,s)$ and ${\tt fibonacci}(r',s')$) 
in the original program are contracted
into one auxiliary function invocation (${\tt fibonacci}^{\it aux}(r,s,r',s')$). As a 
consequence of this step,  the creations of names performed by different invocations 
are contracted to names created by one invocation. This leads to
merging dependencies, which, in turn, reduces the precision of the analysis.
(As discussed in Section~\ref{sec.introduction}, a cardinality argument prevents the inaccuracies introduced 
by $\stackrel{{\tt pl} \mapsto {\tt l}}{\Longmapsto}_3$ from being totally eliminated.)

As far as the correctness of the transformations in Table~\ref{tab.trans} is
concerned, we begin by defining a correspondence between states of a pseudo-linear program 
and those of a linear one. We focus on $\stackrel{{\tt pl} \mapsto {\tt l}}{\Longmapsto}_3$
because the proofs of the correctness of the other transformations are straightforward.

\begin{definition}\label{def.lin}
Let ${\cal L}_2$ be the linear program returned by the 
Transformation 3 of Table~\ref{tab.trans} applied to 
${\cal L}_1$. 
A state $\lam{\por{V}_1,\,\P_1}$ of ${\cal L}_1$ is 
\emph{linearized to} a state $\lam{\por{V}_2,\,\P_2}$ of ${\cal L}_2$, 
written $\lam{\por{V}_1,\,\P_1} \linearized \lam{\por{V}_2,\,\P_2}$, 
if there exists a surjection $\sigma$ such that:
\begin{enumerate}
\item 
$if (\obj,\objb)\in \por{V}_1$ then $(\sigma(\obj),\sigma(\objb))\in \por{V}_2$.
\item 
if $\flatt{\P_1}=\T_1 \seq \cdots \seq \T_m$ and $\flatt{\P_2}=\T_1' \seq \cdots \seq \T_n'$, then for every $1\leq i\leq m$, there exists $1\leq j\leq n$, such that 
$ \sigma(\T_i) \in \T'_j$;
\item if $\M(\wt{x}_1)\in\P_1$ then either (1) 
$\M(\sigma(\wt{x}_1))$ in $\P_2$ or 
(2) there are $\M(\wt{x}_2)\cdots\M(\wt{x}_k)$ in $\P_1$ and $\M^{\it aux}(\wt{y}_1, \cdots, \wt{y}_h)$ in $\P_2$ such that, for every $1 \leq k' \leq
k$ there exists $h'$ with $\sigma(\wt{x}_{k'}) = \wt{y}_{h'}$;
\end{enumerate}
\end{definition}

In the following lemma we use the notation $\CP[\P_1] \cdots [\P_n]$ defined in terms 
of standard lam context by
$(\cdots ((\CP[\P_1])[\P_2]) \cdots )[\P_n]$.

\begin{lemma}\label{th:pl2l}
Let $\lam{\por{V}_1,\,\P_1} \linearized \lam{\por{V}_2,\,\P_2}$.
Then, 
$\lam{\por{V}_2,\,\P_2}\lred{}%
\lam{\por{V}'_2,\,\P'_2}$ implies
there exists 
$\lam{\por{V}_1,\,\P_1}\lred{}^{*}%
\lam{\por{V}'_1,\,\P'_1}$
such that $\lam{\por{V}'_1,\,\P'_1} \linearized \lam{\por{V}'_2,\,\P'_2}$
\end{lemma}
\begin{proof}
\emph{Base case.}
Initially $\P_1=\P_2$ because the main lam is not affected by the transformation. 
Therefore the first step can only be an invocation of a standard function 
belonging to both programs.
We have two cases:
\begin{enumerate}
\item 
the function was linear already in the original program, thus it was not modified by the 
transformation. In this case the two programs performs the same reduction step and 
end up in the same state.
\item the function has been \emph{linearized} by the transformation. In this case the 
invocation at the linear side will reduce to an invocation of an {\it aux}-function and 
it will not produce new pairs nor new names. The corresponding reduction in 
$\lam{\por{V}_1,\P_1}$ is a zero-step reduction. 
It is easy to verify that  $\lam{\por{V}_1,\P_1} \linearized 
\lam{\por{V}_2',\P_2'}$.
\end{enumerate}

\medskip

\noindent
\emph{Inductive case.}
We consider only the case in which the selected function is an {\it aux}-function. The other case is as in the base case.
Let 
\[
\begin{array}{l}
\ilam{\mathbb{V}_1^{(n)}, \, \CP_1^{(n)}[\M(\wt{v}_1)]\cdots[\M(\wt{v }_k)]
}{} 
\\
\qquad \qquad \qquad \qquad \linearized \quad
\ilam{\mathbb{V}_2^{(n)}, \, \CP_2^{(n)}[\M^{\it aux}(\wt{u}_1, \cdots, \wt{u}_h)] }{}
\end{array}
\]
Without loss of generality we can assume that $\CP_1^{(n)}$ does not contain other invocations to $\M$ and the ``linearized to''
relationship makes $\M(\wt{v}_1)\sparop \cdots \sparop \M(\wt{v}_k)$ correspond to 
$\M^{\it aux}(\wt{u}_1, \cdots, \wt{u}_h)$.%
Then we have
\[
\begin{array}{@{\!\!}l}
	\ilam{\mathbb{V}_2^{(n)}, \, \CP_2^{(n)}[\M^{\it aux}(\wt{u}_1, \cdots, \wt{u}_h)]
	}{} \; \lred{} 
	\\	
	\ilam{\mathbb{V}_2^{(n)} \transclosure \wt{u}_1, \cdots, \wt{u}_h {\lessc} 
	\wt{w}, \; 
	\CP_2^{(n)}[\P_{\M^{\it aux }}\subst{\wt{w}}{\wt{z}}\subst{\wt{u}_1, \cdots, \wt{u}_h}{\wt{y}_1 ,\cdots, \wt{y}_h}] }{}
\end{array}
\]
where, $\M^{\it aux}(\wt{y}_1, \cdots , \wt{y}_h)= \P_{\M^{\it aux}}$, 
$\var{\P_{\M^{\it aux}}}\setminus\wt{y}_1 \cdots \wt{y}_h = \wt{z}$ and
$\wt{w}$ are fresh names. 
By construction,
\[
\P_{\M^{\it aux}} = \M^{\it aux}(\wt{y'_1}\subst{\wt{y_1}}{\wt{y}}, \cdots , \wt{y'_k}
\subst{\wt{y_k}}{\wt{y}})\, \sparop\, \prod_{i \in 1..k} (\flatm{\M}{\P_{\M}} \subst{\wt{y_i}}{\wt{y}})
\]
where $\M(\wt{y}) = \CP_\M[\M(\wt{y'_1})] \cdots [\M(\wt{y'_k})] = \P_\M$ and
$\M \notin \CP_\M$.

The corresponding reduction steps of $\ilam{\mathbb{V}_1^{(n)}, \, \CP_1^{(n)}[\M(\wt{v}_1)]\cdots[\M(\wt{v}_k)]}{}$ are the following ones:

\[
\begin{array}{@{\!\!}l}
	\ilam{\mathbb{V}_1^{(n)}, \,  \CP_1^{(n)}[\M(\wt{v}_1)]\cdots[\M(\wt{v}_k)]}{} 
	\; \lred{\M(\wt{v}_1)}
\cdots\lred{\M(\wt{v}_k)} 
\\
\ilam{\mathbb{V}_1^{(n)} \transclosure \wt{v}_1 {\lessc} \wt{w}_1 \transclosure \cdots \transclosure \wt{v}_k {\lessc} \wt{w}_k, \; 
	 \CP_1^{(n)}[\P_\M\subst{\wt{v}_1}{\wt{y}}]\cdots[\P_\M\subst{\wt{v}_k}{\wt{y}}]}{}  
\end{array}
\]
and $\wt{w}_i$ are the fresh names created by the invocation $\M(\wt{v}_i)$, $1\leq i\leq k$.
We need to show that:
\[
\begin{array}{l}
	\ilam{\mathbb{V}_1^{(n)} \transclosure \wt{v}_1 {\lessc} \wt{w}_1 \transclosure \cdots \transclosure \wt{v}_k {\lessc} \wt{w}_k, \; 
	 \CP_1^{(n)}[\P_\M\subst{\wt{v}_1}{\wt{y}}]\cdots[\P_\M\subst{\wt{v}_k}{\wt{y}}]}{}  
\\
\qquad \qquad \qquad \linearized
\\
	\ilam{\mathbb{V}_2^{(n)} \transclosure \wt{u}_1, \cdots, \wt{u}_h {\lessc} 
	\wt{w}, \; 
	\CP_2^{(n)}[\P^{\it aux}_\M]}{}
\end{array}
\]
where $\P^{\it aux}_\M = \M^{\it aux}(\wt{y'_1}\subst{\wt{u_1}}{\wt{y}}, \cdots , \wt{y'_k}
\subst{\wt{u_k}}{\wt{y}})\, \sparop\,  \P^{\it aux}$ 
and
$\P^{\it aux} = \prod_{i \in 1..k} (\flatm{\M}{\P_{\M}} \subst{\wt{u_i}}{\wt{y}})\subst{\wt{w}}{\wt{z}}$.
To this aim we observe that:
\begin{itemize}

\item   
for every $1 \leq k' \leq k$ there exists $h'$ such that $\sigma(\wt{v}_{k'}) = \wt{u}_{h'}$; moreover $\wt{w}=\sigma(\wt{w_1})=\cdots=\sigma(\wt{w_k})$. This satisfies condition {\it 1} of Definition~\ref{def.lin};

\item if $(a,b)\in\P_\M\subst{\wt{v}_i}{\wt{y}}$, with $a,b\in\wt{w}_i,\wt{v}_i$, then  $(\sigma(a),\sigma(b))\in\flatm{\M}{\P_\M}\subst{\wt{u}_i}{\wt{y}}\subst{\wt{w}}{\wt{z}}$, being $\sigma$ defined as in the previous item, therefore $\sigma(a),\sigma(b)\in\wt{w},\wt{u}_i$. Notice that, due to the $\prod_{i \in 1..k}$ composition in the body of $\M^{\it aux}$, two pairs sequentially composed in $\P_\M$ may end up in parallel (through $\sigma$). The converse never happens. Therefore condition {\it 2} of Definition~\ref{def.lin} is satisfied.

\item 
if $\N(\wt{a})\in\P_\M$ we can reason as in the previous item. We notice that 
$\prod_{i \in 1..k} (\flatm{\M}{\P_{\M}} \subst{\wt{u_i}}{\wt{y}})\subst{\wt{w}}{\wt{z}}$ may contain 
function invocations $\N(\wt{u})$ that have no counterpart (through $\sigma$) in
$\P_\M\subst{\wt{v}_i}{\wt{y}}$. We do not have to mind about them because the lemma
guarantees the converse containment.

\item 
in $\P_\M\subst{\wt{v}_i}{\wt{y}}$ we have $k$ new invocations of $\M(\wt{b}_{i,1})\cdots\M(\wt{b}_{i,k})$, 
where $\wt{b}_{i,j}=\wt{y}'_{j}\subst{\wt{v}_j}{\wt{y}}\subst{\wt{w}_j}{\wt{z}}$. 
Therefore in the pseudolinear lam we have $k^2$ invocations of $\M$, while in the corresponding linear lam we find just one invocation of 
$\M^{\it aux}(\wt{y'_1}\subst{\wt{u_1}}{\wt{y}}\subst{\wt{w}}{\wt{z}}, \cdots ,$ $\wt{y'_k}
\subst{\wt{u_k}}{\wt{y}}\subst{\wt{w}}{\wt{z}})$. 
The surjection $\sigma$ is such that $(\wt{y'_j}\subst{\wt{u_j}}{\wt{y}}\subst{\wt{w}}{\wt{z}}=\sigma(\wt{b}_{1,j})=\cdots=\sigma(\wt{b}_{k,j})$,  with $1\leq j \leq k$.
This, together with the previous item, satisfies condition {\it 3} of Definition~\ref{def.lin}.
\end{itemize}
\end{proof}

\begin{lemma}\label{th:pl2lbis}
Let $\lam{\por{V}_1,\,\P_1} \linearized \lam{\por{V}_2,\,\P_2}$ and
$\lam{\por{V}_1,\,\P_1}\lred{}^{*}
\lam{\por{V}'_1,\,\P'_1}$.
Then there are $\lam{\por{V}'_1,\,\P'_1}\lred{}^{*}
\lam{\por{V}''_1,\,\P''_1}$ and 
$\lam{\por{V}_2,\,\P_2}$ $\lred{}^{*}
\lam{\por{V}'_2,\,\P'_2}$
such that $\lam{\por{V}''_1,\,\P''_1} \linearized \lam{\por{V}'_2,\,\P'_2}$
\end{lemma}

\begin{proof}
A straightforward induction on the length of $\lam{\por{V}_1,\,\P_1}\lred{}^{*}
\lam{\por{V}'_1,\,\P'_1}$. In the inductive step, we need to expand the recursive
invocations ``at a same level'' in order to mimic the behavior of functions
$\M^{\it aux}$.
\end{proof}

\begin{theorem}
\label{thm.transformation-pl-l}
Let ${\cal L}_1$ be a pseudo-linear program and 
${\cal L}_2$ be the result of the transformations in Table~\ref{tab.trans}.
If a saturated state of ${\cal L}_2$ has no circularity then no state
of ${\cal L}_1$ has a circularity.
\end{theorem}

\begin{proof}
The transformations $\stackrel{{\tt pl} \mapsto {\tt l}}{\Longmapsto}_1$ and $\stackrel{{\tt pl} \mapsto {\tt l}}{\Longmapsto}_2$ perform expansions and do not introduce
inaccuracies. By Lemma~\ref{th:pl2l}, for every $\lam{\por{V}_2,\,\P_2}$ reached by evaluating 
${\cal L}_2$, there is $\lam{\por{V}_1,\,\P_1}$ that is reached by evaluating 
${\cal L}_1$ such that $\lam{\por{V}_1,\,\P_1} \linearized \lam{\por{V}_2,\,\P_2}$.
This guarantees that every circularity in $\lam{\por{V}_1,\,\P_1}$ is also present
in $\lam{\por{V}_2,\,\P_2}$. We conclude by Lemma~\ref{th:pl2lbis} and
Theorem~\ref{thm.invariance}.
\end{proof}

We observe that,  our analysis returns that the ${\tt fibonacci}$ program
is circularity-free.

\subsection{The general case}
\label{sec.generalcase}
In non-pseudo-linear recursive programs, more than one mutual recursive function may have 
several recursive histories. 
The transformation 
$\stackrel{{\tt npl} \mapsto {\tt pl}}{\Longmapsto}$ in Table~\ref{tab.t4}
takes a non-pseudo-linear recursive program and returns a program where the 
``non-pseudo-linearity'' is simpler. Repeatedly applying the transformation,
at the end, one obtains a pseudo-linear recursive program.  

More precisely, let $\bigl(\M_1(\wt{x_1}) = \P_{1}, \cdots , 
\M_\ell(\wt{x_\ell}) = \P_{\ell}, \P \bigr)$ be a non-pseudo-linear recursive program. Therefore,
there are at least two functions with more than one recursive history. One of this 
function is $\M_j$, which is the one that is being explored by the rule 
$\stackrel{{\tt npl} \mapsto {\tt pl}}{\Longmapsto}$. Let also $\M_i$ be another 
function such that $\closure{\M_j} = \closure{\M_i}$ (this $\M_i$ must exists otherwise 
the program would be already pseudo-linear recursive). These constraints are  
those listed in the first line of the premises of the rule.
The idea of this transformation is to defer the invocations
of the functions in $\{ \head{\alpha_1 \M_j}, \cdots ,$ 
$\head{\alpha_{h+1}\M_j} \} \setminus \M_i$, i.e., the functions different from $\M_i$ 
that can be invoked within $\M_j$'s body, to the body of the function $\M_i$. The 
meaning of the second and third lines of the premises of the rule  is to 
identify the $p_k$ different invocations of these $m$ functions ($k \geq m$). 
Notice that every $\alpha_1,\cdots,\alpha_{h+1}$ could be empty, meaning that 
$\M_j$ is directly called.
At this point, what we need to do is (1)  to  store the arguments of each invocation 
of $\M_{i_1},\cdots,\M_{i_m}$ into those of 
an invocation of $\M_i$ -- actually,  a suitable tuple of them, thus the arity of 
$\M_i$ is augmented
correspondingly -- and (2) to perform suitable expansions in the body of $\M_i$.
In order to augment the arguments of the invocations of $\M_i$ that occur 
in the other parts of 
the program, we use the auxiliary rule $\transflam{\M_i,n}$ that extends every 
invocation of $\M_i$ with $n$ additional arguments that are always fresh names. 
The fourth line of the premises calculates the number $n$ of additional 
arguments, based on the number of arguments of the functions that are going 
to be moved into $\M_i$'s body.
The last step, described in the last line of the premises of the rule, is to 
replace the invocations of the functions  $\M_{i_1},\cdots,\M_{i_m}$ with 
invocations of $\M_i$. Notice that, in each invocation, the position of 
the actual arguments is different. In the body of $\M_i$, after the 
transformation, the invocations of those functions will be performed passing 
the right arguments.
\begin{table*}[t]
{\footnotesize
\[
\begin{array}{c}
\bigfract{
\begin{array}{c}
\M \notin \CP
\qquad
\wt{z_1}, \cdots , \wt{z_m} \mbox{\emph{ are $n$-tuple of fresh names}}
\end{array}
}{
\CP[\M(\wt{u_1})] \cdots [\M(\wt{u_m})] \transflam{\M,n} 
\CP[\M(\wt{u_1},\wt{z_1})] \cdots [\M(\wt{u_m},\wt{z_m})]
}
\\
\\
\\
\bigfract{
	\begin{array}{c}
		\rechis{\M_j} = \{ \M_j \M_i \alpha_0 , \M_j \alpha_1 , \cdots , 
		\M_j \alpha_{h+1}\}
	\qquad  
	\M \in \M_i\alpha_0 \qquad \sharp(\rechis{\M}) > 1
	\\
	\{\M_{i_1}, \cdots , \M_{i_m}\} = \{ \head{\alpha_1 \M_j}, \cdots , 
	\head{\alpha_{h+1}\M_j} \} \setminus \M_i
	\\ 
    \qquad
	\P_j = \CP[\M_{p_1}(\wt{u_1})] \cdots [\M_{p_k}(\wt{u_{k}})] 
	\qquad \{ \M_{p_1}, \cdots , \M_{p_k} \} = \{ \M_{i_1}, \cdots , \M_{i_m}\} \qquad
	\M_{i_1}, \cdots , \M_{i_m} \notin \CP
	\\
	n = \sharp(\wt{u_1} \cdots \wt{u_k}) \qquad
	(\P_h \transflam{\M_i,n} \P_h')^{h \in \{ 1, \cdots, \ell+1 \}}
	\qquad 
	\P_j' = \CP'[\M_{i_1}(\wt{u_1})] \cdots [\M_{i_m}(\wt{u_{k}})]
	\\
	\wt{z_1^1}, \cdots , \wt{z_{k}^1}, \cdots ,  \wt{z_1^k},  \cdots , \wt{z_{k}^k},
	 \wt{z_1},  \cdots , \wt{z_{k}},
	\mbox{ are fresh} 
	\\
	\P_j'' = \CP'
	[\M_{i}(\wt{z_1^1},
	\wt{u_1},  \wt{z_2^1},\cdots , \wt{z_{k}^1})] \cdots 
	[\M_{i}(\wt{z_1^k},  \cdots , \wt{z_{k}^k},\wt{u_{k}})]
	 \end{array}
}{
\begin{array}{l}
\bigl( \M_1(\wt{x_1}) = \P_1,\cdots \M_i(\wt{x_i}) = \P_i, \cdots ,
\M_j(\wt{x_j}) = \P_j, \cdots ,
\M_\ell(\wt{x_\ell}) = \P_\ell, \P_{\ell+1} \bigr)
\; \stackrel{{\tt npl} \mapsto {\tt pl}}{\Longmapsto}
\\
\bigl( \M_1(\wt{x_1}) = \P_1',\cdots \M_i(\wt{x_i}, \wt{z_1}, \cdots , \wt{z_k}) = \P_i'\sparop 
	(\prod_{q \in 1..k} \M_{p_q}(\wt{z_q})), \cdots , \M_j(\wt{x_j}) = \P_j'', \cdots ,
\M_\ell(\wt{x_\ell}) = \P_\ell', \P'_{\ell+1} \bigr)
\end{array}
}
\end{array}
\] }
\caption{\label{tab.t4} Non-pseudo-linear to pseudo-linear transformation}
\end{table*}
For example, the $\M''\N''$-program 

\smallskip
 
$\begin{array}{ll}
\bigl( \; & \M''(x,y) =  (x,z)\sparop\M''(y,z) \seq \N''(y,x) \; ,
\\
 &\N''(x,y) = (y,x)\sparop\M''(y,z) \sparop \N''(z,x) \;,
 \\
 & \M''(x_1,x_2) \quad
 \bigr)
\end{array}$

\smallskip

is rewritten into

\smallskip 

$\begin{array}{@{\!}ll}
\bigl( & \M''(x,y) =  (x,z)\sparop\N''(x',y',y,z) \seq \N''(y,x,z',z'') \; ,
\\
 &\N''(x,y,u,v) = (y,x)\sparop\M''(y,z) \sparop \N''(z,x,x',y') \sparop \M''(u,v) \;,
 \\
 & \M''(x_1,x_2) \quad
 \bigr) \; .
\end{array}$

\smallskip

The invocation $\M''(y,z)$ is moved into the body of $\N''$. The function $\N''$ has 
an augmented arity, 
so that its first two arguments
refer to the arguments of the invocations of $\N''$ in the original program, and the last two arguments refer to the invocation of $\M''$. 
Looking at the body of $\N''$, the unchanged part (with the augmented arity of $\N''$)
 covers the first two 
arguments; whilst the 
last two arguments are 
only used for a new invocation of $\M''$.

The correctness of 
$\stackrel{{\tt npl} \mapsto {\tt pl}}{\Longmapsto}$ is demonstrated in a similar
way to the proof of the correctness of $\stackrel{{\tt pl} \mapsto 
{\tt l}}{\Longmapsto}_3$. We begin by defining a correspondence between states 
of a non-pseudo-linear program 
and those of a pseudo-linear one. 

\begin{definition}\label{def.pslin}
Let ${\cal L}_2$ be the pseudo-linear program returned by the 
transformation  of Table~\ref{tab.t4} applied to 
${\cal L}_1$. 
A state $\lam{\por{V}_1,\,\P_1}$ of ${\cal L}_1$ is 
\emph{pseudo-linearized to} a state $\lam{\por{V}_2,\,\P_2}$ of ${\cal L}_2$, 
written $\lam{\por{V}_1,\,\P_1} \pslinearized \lam{\por{V}_2,\,\P_2}$, 
if there exists a surjection $\sigma$ such that:
\begin{enumerate}
\item 
$if (\obj,\objb)\in \por{V}_1$ then $(\sigma(\obj),\sigma(\objb))\in \por{V}_2$.
\item 
if $\flatt{\P_1}=\T_1 \seq \cdots \seq \T_m$ and $\flatt{\P_2}=\T_1' \seq \cdots \seq \T_n'$, then for every $1\leq i\leq m$, there exists $1\leq j\leq n$, such that 
$\sigma( \T_i) \in \T'_j$;
\item if $\M(\wt{x})\in\P_1$ then either (1) 
$\M(\sigma(\wt{x}))$ in $\P_2$ or 
(2) there is $\M(\wt{y}_1\cdots\wt{y}_k)$ in $\P_2$ such that, for some $1 \leq i \leq
k$, $\sigma(\wt{x}) = \wt{y}_{i}$;

\end{enumerate}
\end{definition}

We use the same notational convention for contexts as in Lemma~\ref{th:pl2l}.

\begin{lemma}\label{th:nl2pl}
Let $\lam{\por{V}_1,\,\P_1} \pslinearized \lam{\por{V}_2,\,\P_2}$.
Then, 
$\lam{\por{V}_1,\,\P_1}\lred{}%
\lam{\por{V}'_1,\,\P'_1}$ implies
there exists 
$\lam{\por{V}_2,\,\P_2}\lred{}^{+}%
\lam{\por{V}'_2,\,\P'_2}$
such that $\lam{\por{V}'_1,\,\P'_1} \pslinearized \lam{\por{V}'_2,\,\P'_2}$
\end{lemma}
\begin{proof}
\emph{Base case.}
$\P_1$ is the main lam of the nonlinear program, and $\P_2$ its pseudolinear transformation.
\[
\P_1=\CP_1[\M_1(\wt{u}_1)]\cdots[\M_m(\wt{u}_k)],
\]
where $\CP_1$ does not contain any other function invocations, and $m\leq k$, meaning that some of the $\M_i$, $1\leq i\leq m$, can be invoked more than once on different parameters.

After the transformation, $\P_2$ contains the same pairs as $\P_1$ and the same function invocations, but with possibly more arguments:
\[
\P_2=\CP_1[\M_1(\wt{u}_1,\wt{z}_1)]\cdots[\M_m(\wt{u}_k,\wt{z}_k)].
\]
Notice that some of the $\wt{z}_j$, $1\leq j\leq k$, may be empty if the corresponding function has not been expanded during the transformation.
Moreover $\por{V}_1$ and $\por{V_2}$ contains only the identity relations on the arguments, so we have $\por{V}_1\subseteq\por{V}_2$.
Therefore, all conditions of definition~\ref{def.pslin} are trivially verified.

\medskip

\noindent
\emph{Inductive case.}
We have
\[
\P_1=\CP_1[\M_1(\wt{u}_1)]\cdots[\M_m(\wt{u}_k)],
\]
where $\CP_1$ does not contain any other function invocations, and $m\leq k$, meaning that some of the $\M_i$, $1\leq i\leq m$, can be invoked more than once on different parameters.

We have
\[
\P_2=\CP_2[\M_1(\wt{u}_1,\wt{z}_1)]\cdots[\M_m(\wt{u}_k,\wt{z}_k)].
\]
where $\CP_2$ may contain other function invocations, but by inductive hypothesis we know that Definition~\ref{def.pslin} is verified. In particular condition \emph{3} guarantees that at least the invocations of $\M_1,\ldots,\M_m$, with suitable arguments, are in $\P_2$.

Now, let us consider the reduction
\[
\lam{\por{V}_1,\,\P_1}\lred{}\lam{\por{V}'_1,\,\P'_1}.
\]
Without loss of generality, we can assume the reduction step performed an invocation of function $\M_1(\wt{u}_1)$.

We have different cases:
\begin{enumerate}
  \item the function's lam $\P_{\M_1}$ has not been modified by the transformation. In this case the result follows trivially.
  \item the function's lam $\P_{\M_1}$ has been affected only in that some function invocations in it have an updated arity. Meaning that
it was only trasformed by $\transflam{\N,l}$, for some $\N$ and $l$, as a side effect of other function expansions.
It follows that $\flat(\P_{\M_1})=\flat(\P'_{\M_1})$, where $\P'_{\M_1}$ is the body of $\M_1$ after the transoformation has been applied. This satisfies  condition \emph{2} of Definition~\ref{def.pslin}.
Those function invocations that have not been modified satisfy trivially the condition \emph{3} of Definition~\ref{def.pslin}. Regarding the other function invocations we have, by construction, that if $\N(\wt{x})\in\P_{\M_1}$ then $\N(\wt{x},\wt{y})\in\P'_{\M_1}$, where $\wt{y}$ are fresh names. This satisfies condition \emph{3} of Definition~\ref{def.pslin}, as well.
As for condition \emph{1}, we have
\[
\por{V}'_1=\por{V}_1\transclosure(\wt{u}_1<\wt{w}_1),
\]
where $\wt{w}$ are fresh names created in $\P_{\M_1}$, and
\[
\por{V}'_2=\por{V}_2\transclosure(\wt{u}_1,\wt{z_1}<\wt{w}_1,\wt{y}_1,\cdots,\wt{y}_s),
\]
where $\wt{y}_1,\cdots,\wt{y}_s$ are the fresh names augmenting the function arities within $\P'_{\M_1}$.
We choose the same fresh names $\wt{w}_1$ and condition \emph{1} is satisfied.
\item the function's lam $\P_{\M_1}$ has been subject of the expansion of a function.
Let 
\[
\P_{\M_1}=\CP_{\M_1}[\N_1(\wt{v}_1)]\cdots[\N_h(\wt{v}_n)],
\]
where $\CP_{\M_1}$ contains only pairs, then, assuming without loss of generality that $\N_1$ was expanded:
\[
\P'_{\M_1}=\CP_{\M_1}[\N_1(\wt{v}_1,\wt{z}^1_1,\ldots,\wt{z}^1_r)]\cdots[\N_1(\wt{v}_n,\wt{z}^r_1,\ldots,\wt{z}^r_r)],
\]
where $r$ is obtained by subtracting from the number of invocations $n$ the number of occurrences of invocations of $\N_1$ in $\P'_{\M_1}$.

Now, the psedulinear program has to perform the $r$ invocations of $\N_1$ that were not present in the original program, since they have been replaced $r$ invocations of $\N_2\cdots\N_h$, in order to reveal the actual invocations $\N_2\cdots\N_h$ that has been delegated to $\N_1$ body.
By construction, the arguments of the invocations where preserved by the transformation, so that if $\N_2(\wt{x})$ is produced by reduction of the nonlinear program, then the pseudolinear program will produce $\N_2(\wt{x},\wt{y})$, with $\wt{y}$ fresh and possibily empty. This satisfy condition \emph{3} of Definition~\ref{def.pslin}.

However the body of $\N_1$ may have been transformed in a similar way by expanding another method, let us say $\N_2$. Then all the invocations of $\N_2$ in $\N_1$'s body that corresponds to the previously delegated function invocations $\N_2\cdots\N_h$ have to be invoked as well.
This procedure has to be iterated until all the corresponding invocations are encountered. Each step of reduction will produce spurious pairs and function invocations, but all of these will be on different new names.
\end{enumerate}  
\end{proof}

\begin{lemma}\label{th:nl2plbis}
Let $\lam{\por{V}_1,\,\P_1} \pslinearized \lam{\por{V}_2,\,\P_2}$ and
$\lam{\por{V}_1,\,\P_1}\lred{}^{*}
\lam{\por{V}'_1,\,\P'_1}$.
Then there are $\lam{\por{V}'_1,\,\P'_1}\lred{}^{*}
\lam{\por{V}''_1,\,\P''_1}$ and 
$\lam{\por{V}_2,\,\P_2}$ $\lred{}^{*}
\lam{\por{V}'_2,\,\P'_2}$
such that $\lam{\por{V}''_1,\,\P''_1} \pslinearized \lam{\por{V}'_2,\,\P'_2}$
\end{lemma}

\begin{proof}
A straightforward induction on the length of $\lam{\por{V}_1,\,\P_1}\lred{}^{*}
\lam{\por{V}'_1,\,\P'_1}$. 
\end{proof}

\medskip

Every preliminary result is in place for the correctness of the transformation
$\stackrel{{\tt npl} \mapsto {\tt pl}}{\Longmapsto}$.

\begin{theorem}
\label{thm.transformation-npl-pl}
Let ${\cal L}_1$ be a non-pseudo-linear program and 
${\cal L}_2$ be the result of the transformations in Table~\ref{tab.t4}.
If ${\cal L}_2$ is circularity-free then
${\cal L}_1$ is circularity-free.
\end{theorem}

\begin{proof}
By Lemma~\ref{th:nl2pl}, for every $\lam{\por{V}_1,\,\P_1}$ reached by evaluating 
${\cal L}_1$, there is $\lam{\por{V}_2,\,\P_2}$ that is reached by evaluating 
${\cal L}_2$ such that $\lam{\por{V}_1,\,\P_1} \pslinearized \lam{\por{V}_2,\,\P_2}$.
This guarantees that every circularity in $\lam{\por{V}_1,\,\P_1}$ is also present
in $\lam{\por{V}_2,\,\P_2}$. We conclude by Lemma~\ref{th:nl2plbis}. 
\end{proof}

\end{document}